\documentclass[11pt]{article}

\usepackage{times}
\usepackage{amsmath}
\usepackage{algorithm}
\usepackage{amsthm}
\usepackage{graphicx,epsfig}
\usepackage{caption}
\usepackage[noend]{algpseudocode}
\usepackage{fullpage}

\newtheorem{theorem}{Theorem}[section]
\newtheorem{corollary}[theorem]{Corollary}
\newtheorem{lemma}[theorem]{Lemma}
\newtheorem{observation}[theorem]{Observation}

\newtheorem{claim}[theorem]{Claim}

\newtheorem{assumption}[theorem]{Assumption}

\newtheorem{invariant}[theorem]{Invariant}
\newtheorem{property}[theorem]{Property}

\usepackage{hyperref}

\newenvironment{lp*}{\begin{equation*}  \begin{array}{lll}}{\end{array}\end{equation*}}

\begin{document}

\title{Deterministically Maintaining a $(2+\epsilon)$-Approximate Minimum Vertex Cover  in $O(1/\epsilon^2)$ Amortized Update Time}
\author{Sayan Bhattacharya\thanks{University of Warwick, Coventry, UK. {\tt Email:} S.Bhattacharya@warwick.ac.uk} \and Janardhan Kulkarni\thanks{Microsoft Research, Redmond, USA. {\tt Email:} janardhan.kulkarni@gmail.com}}

\begin{titlepage}
	\maketitle
	\pagenumbering{roman}
\begin{abstract}
We consider the problem of maintaining an (approximately) minimum vertex cover in an $n$-node  graph $G = (V, E)$ that is getting updated dynamically via a sequence of edge insertions/deletions. We show how to maintain a $(2+\epsilon)$-approximate minimum vertex cover, {\em deterministically}, in this setting in $O(1/\epsilon^2)$ amortized update time. 

Prior to our work, the best known deterministic algorithm for maintaining a $(2+\epsilon)$-approximate minimum vertex cover was due to Bhattacharya, Henzinger and Italiano [SODA 2015]. Their algorithm has an update time of $O(\log n/\epsilon^2)$. Recently, Bhattacharya, Chakrabarty, Henzinger [IPCO 2017] and Gupta, Krishnaswamy, Kumar, Panigrahi [STOC 2017] showed how to maintain an $O(1)$-approximation in $O(1)$-amortized update time for the same problem.  Our result gives an {\em exponential} improvement over the update time of Bhattacharya et al. [SODA 2015], and nearly matches the performance of the {\em randomized} algorithm of Solomon [FOCS 2016] who gets an approximation ratio of $2$ and an expected amortized update time of $O(1)$.

We derive our result by analyzing, via a novel technique, a variant of the algorithm by Bhattacharya et al. We consider an idealized setting where the update time of an algorithm can take any arbitrary fractional value, and use  insights from this setting to come up with an appropriate potential function. Conceptually, this framework mimics the idea of an LP-relaxation for an optimization problem. The difference is that instead of relaxing an integral objective function, we relax the update time of an algorithm itself. We believe that this technique will find further applications in the analysis of dynamic algorithms.
\end{abstract}
	\newpage
	\setcounter{tocdepth}{2}
	\tableofcontents
\end{titlepage}

\newpage

\newcommand{\E}{\mathcal{E}}

\part{Extended Abstract}
\label{part:main}
\pagenumbering{gobble}

\newpage
\pagenumbering{arabic}

\section{Introduction}
\label{sec:intro}

Consider an undirected graph $G = (V, E)$ with $n= |V|$ nodes, and suppose that we have to compute an (approximately) minimum vertex cover\footnote{A vertex cover in $G$ is a subset of nodes $S \subseteq V$ such that every edge $(u, v) \in E$ has at least one endpoint in $S$.}  in $G$. This problem is well-understood in the static setting. There is a simple linear time greedy algorithm that returns a {\em maximal matching}\footnote{A matching in $G$ is a subset of edges $M \subseteq E$ such that no two edges in $M$ share a common endpoint. A matching $M$ is maximal if for every edge $(u, v) \in E \setminus M$, either $u$ or $v$ is matched in $M$.} in $G$. Let $V(M) \subseteq V$ denote the set of nodes that are matched in $M$. Using the duality between maximum matching and minimum vertex cover, it is easy to show that the set $V(M)$ forms a $2$-approximate minimum vertex cover in $G$. Accordingly, we can compute a $2$-approximate minimum vertex cover  in linear time. In contrast, under the Unique Games Conjecture~\cite{KhotR03}, there is no polynomial time  $(2-\epsilon)$-approximation algorithm for minimum vertex cover for any $\epsilon > 0$. In this paper, we consider the problem of maintaining an (approximately) minimum vertex cover in a {\em dynamic} graph, which gets {\em updated} via a sequence of edge insertions/deletions. The time taken to handle the insertion or deletion of an edge is called the {\em update time} of the algorithm. The goal is to design a dynamic algorithm with small approximation ratio whose update time is significantly faster than the trivial approach of recomputing the solution {\em from scratch} after every update.

A naive approach for this problem will be to maintain a maximal matching $M$ and the set of matched nodes $V(M)$ as follows. When an edge $(u, v)$ gets inserted into $G$, add the edge $(u, v)$ to the matching iff both of its endpoints $u, v$ are currently unmatched. In contrast, when an edge $(u, v)$ gets deleted from $G$, first check if the edge $(u, v)$ was  matched in $M$ just before this deletion. If yes, then remove the edge $(u, v)$ from $M$ and try to {\em rematch} its endpoints $x \in \{u, v\}$. Specifically, for every endpoint $x \in \{u, v\}$, scan through all the edges $(x, y) \in E$ incident on $x$ till an edge $(x, y)$ is found  whose other endpoint $y$ is currently unmatched, and at that point  add the edge $(x, y)$ to the matching $M$ and stop the scan. Since a node $x$ can have $\Theta(n)$ neighbors, this  approach leads to an update time of $\Theta(n)$.

Our main result is stated in Theorem~\ref{th:main}. Note that the amortized update time\footnote{Following the standard convention in dynamic algorithms literature,  an algorithm has $O(\alpha)$ amortized update time  if starting from a graph $G = (V, E)$ where $E = \emptyset$, it takes $O(t \cdot \alpha)$ time overall to handle any sequence of $t$ edge insertions/deletions in $G$.} of our algorithm is independent of $n$. As an aside, our algorithm  also maintains a $(2+\epsilon)$-approximate maximum {\em fractional} matching as a  dual certificate, deterministically, in $O(\epsilon^{-2})$ amortized update time.
\begin{theorem}
\label{th:main}
For any $0 < \epsilon < 1$, we can maintain a $(2+\epsilon)$-approximate minimum vertex cover in a dynamic graph, deterministically, in $O(\epsilon^{-2})$ amortized update time. 
\end{theorem}

\subsection{Perspective} The first major result on maintaining a small  vertex cover and a large matching in a dynamic graph appeared in STOC 2010~\cite{OnakR10}. By now, there is a large body of work devoted to this topic: both on general graphs~\cite{AddankiS18,ArarCCSW18,BaswanaGS11,BernsteinS16,BernsteinS15,BhattacharyaHI15,BhattacharyaHN17,BhattacharyaHN16,CharikarS18,Gupta,GuptaP13,KopelowitzKPS14,NeimanS13,OnakR10,PelegS16} and on graphs with bounded arboricity~\cite{BernsteinS16,BernsteinS15,KopelowitzKPS14,PelegS16}. These results give interesting tradeoffs between various parameters such as (a) approximation ratio, (b) whether the algorithm is deterministic or randomized, and (c) whether the update time is amortized or worst case. In this paper, our focus is on aspects (a) and (b).  {\em We want to design a dynamic algorithm for minimum vertex cover that is deterministic and has (near) optimal approximation ratio and amortized update time.} In particular, we are {\em not} concerned with  the worst case update time of our algorithm. From this specific point of view, the literature on dynamic vertex cover can be summarized as follows.

\smallskip
\noindent {\bf Randomized algorithms.} Onak and Rubinfeld~\cite{OnakR10} presented a randomized algorithm for maintaining a $O(1)$-approximate minimum vertex cover in $O(\log^2 n)$ expected update time. This bound was improved upon by Baswana, Gupta, Sen~\cite{BaswanaGS11}, and subsequently by Solomon~\cite{Solomon16}, who obtained a $2$-approximation in $O(1)$ expected update time.

\smallskip
\noindent {\bf Deterministic algorithms.} Bhattacharya, Henzinger and Italiano~\cite{b} showed how to deterministically maintain a $(2+\epsilon)$-approximate minimum vertex cover in $O(\log n/\epsilon^2)$ update time. Subsequently,  Bhattacharya, Chakrabarty and Henzinger~\cite{BhattacharyaCH17} and Gupta, Krishnaswamy, Kumar and Panigrahy~\cite{Gupta} gave deterministic dynamic algorithms for this problem with an approximation ratio of $O(1)$ and update time of $O(1)$. The algorithms designed in these two papers~\cite{BhattacharyaCH17,Gupta} extend to the more general problem of dynamic set cover.

\begin{table}[h]
\begin{small}
\begin{center}
\begin{tabular}{||c|c|c|c||}
\hline

{\textbf{Approximation Ratio}} & {\textbf{Amortized Update Time}} & {\textbf{Algorithm}}  & {\textbf{Reference}} \\	\hline	 
          	
\hline\hline
$(2+\epsilon)$ & $O(\log n/\epsilon^2)$  & deterministic &  Bhattacharya et al.~\cite{b} \\ 
\hline
$O(1)$ & $O(1)$ & deterministic & Gupta et al.~\cite{Gupta} and \\
            &             &                       & Bhattacharya et al.~\cite{BhattacharyaCH17} \\
\hline
$2$ & $O(1)$ & randomized & Solomon~\cite{Solomon16} \\
\hline
\end{tabular}
\end{center}
\end{small}
\vspace{-5mm}
\caption{State of the art on dynamic algorithms with fast amortized update times for  minimum vertex cover. }
\label{table:perspective} 
\end{table}

Thus, from our perspective the state of the art results on dynamic vertex cover prior to our paper are summarized in Table~\ref{table:perspective}. Note that the results stated in Table~\ref{table:perspective} are mutually incomparable. Specifically, the algorithms in~\cite{b,Gupta,BhattacharyaCH17} are all deterministic, but the paper~\cite{b} gives near-optimal (under Unique Games Conjecture) approximation ratio whereas the papers~\cite{Gupta,BhattacharyaCH17} give optimal update time. In contrast, the paper~\cite{Solomon16} gives optimal approximation ratio and update time, but the algorithm there is randomized. Our main result as stated in Theorem~\ref{th:main} combines the best of the both worlds, by showing that there is a dynamic algorithm for minimum vertex cover that is simultaneously (a) deterministic, and has (b) near-optimal approximation ratio and (c) optimal update time for constant $0 < \epsilon < 1$. In other words, we get an {\em exponential} improvement in the update time bound in~\cite{b}, without increasing the approximation ratio or using randomization. 

Most of the randomized dynamic algorithms in the literature, including the ones~\cite{AddankiS18,BaswanaGS11,OnakR10,Solomon16} that are relevant to this paper, assume that the adversary is {\em oblivious}. Specifically, this means that the future edge insertions/deletions in the input graph do not depend on the current solution being maintained  by the algorithm. A deterministic dynamic algorithm does not require this assumption, and hence designing deterministic dynamic algorithms for fundamental optimization problems such as minimum vertex cover is an important research agenda in itself. Our result should be seen as being part of this research agenda.

\smallskip
\noindent {\bf Our technique.} A novel and interesting aspect of our techniques is that we {\em relax the notion of update time of an algorithm}. We consider an idealized, continuous world where the update time of an algorithm can take any fractional, arbitrarily small value. We first study the behavior of a natural dynamic algorithm for minimum vertex cover in this idealized world. Using insights from this study, we  design an appropriate potential function for analyzing the update time of a minor variant of the algorithm from~\cite{b} in the {\em real-world}. Conceptually, this framework mimics the idea of an LP-relaxation for an optimization problem. The difference is that instead of relaxing an integral objective function, we relax the update time of an algorithm itself. We believe that this technique will find further applications in the analysis of dynamic algorithms.

\medskip
\noindent {\bf Organization of the rest of the paper.}  In Section~\ref{sec:old:soda15}, we present a summary of the algorithm from~\cite{b}.  A reader already familiar with the algorithm will be able to quickly skim through this section. In Section~\ref{sec:continuous}, we analyze the update time of the algorithm in an idealized, continuous setting.  This sets up the stage for the analysis of our actual algorithm in the real-world. Note that the first ten pages  consists of Section~\ref{sec:intro} -- Section~\ref{sec:continuous}.  For the reader  motivated enough to further explore the paper, we present an overview of our algorithm and analysis in the ``real-world" in Section~\ref{sec:algorithm:new}. The full version of our algorithm, along with a complete analysis of its update time, appears in Part~\ref{part:full}.


\section{The framework of Bhattacharya, Henzinger and Italiano~\cite{b}}
\label{sec:old:soda15}
We henceforth refer to the dynamic algorithm developed in~\cite{b}  as the BHI15 algorithm. In this section, we give a brief overview of the BHI15 algorithm, which is based on a primal-dual approach that simultaneously maintains a fractional matching\footnote{A fractional matching in $G = (V, E)$ assigns a weight $0 \leq w(e) \leq 1$ to each edge $e \in E$, subject to the constraint that the total weight received by any node from its incident edges is at most one. The size of the fractional matching is given by $\sum_{e \in E} w(e)$. It is known that the maximum matching problem is the dual of the minimum vertex cover problem. Specifically, it is known that the size of the maximum fractional matching is at most the size of the minimum vertex cover.} and a vertex cover whose sizes are within a factor of $(2+\epsilon)$ each other. 

\medskip
\noindent {\bf Notations.} Let $0 \leq w(e) \leq 1$ be the weight assigned to an edge $e \in E$. Let $W_y = \sum_{x \in N_y} w(x,y)$ be the total weight received by a node $y \in V$ from its incident edges, where $N_y = \{ x \in V : (x, y) \in E\}$ denotes the set of neighbors of $y$. The BHI15 algorithm maintains a partition of the node-set $V$ into $L+1$ {\em levels}, where $L = \log_{(1+\epsilon)} n$. Let $\ell(y) \in \{0, \ldots, L\}$ denote the level of a node $y \in V$. For any two integers $i, j \in [0, L]$ and any node $x \in V$, let $N_y(i,j) = \{ x \in N_y : i \leq \ell(x) \leq j \}$ denote the set of neighbors of $y$ that lie in between level $i$ and level $j$.  The level of an edge $(x,y) \in E$ is denoted by $\ell(x,y)$, and it is defined to be equal to the maximum level among its two endpoints, that is, we have $\ell(x,y) = \max(\ell(x), \ell(y))$. In the BHI15 framework, the weight of an edge is completely determined by its level. In particular, we have $w(x,y) = (1+\epsilon)^{-\ell(x,y)}$, that is, the weight $w(x, y)$  decreases exponentially with the level  $\ell(x,y)$.

\medskip
\noindent {\bf A static primal-dual algorithm.}
To get some intuition behind the BHI15 framework, consider the following static primal-dual algorithm. The algorithm proceeds in rounds. Initially, before the first round, every node $y \in V$ is at level $\ell(y) = L$ and every edge $(x, y) \in E$ has weight $w(x,y) = (1+\epsilon)^{-\ell(x,y)} = (1+\epsilon)^{-L} = 1/n$.  Since each node has at most $n-1$ neighbors in an $n$-node graph, it follows that $W_y = (1/n) \cdot |N_y|  \leq (n-1)/n$ for all nodes $y \in V$ at this point in time. Thus, we have $0 \leq W_y < 1$ for all nodes $y \in V$, so that the edge-weights $\{w(e)\}$ form a valid fractional matching in $G$ at this stage. We say that a node $y$ is {\em tight} if $W_y \geq 1/(1+\epsilon)$ and {\em slack} otherwise. In each subsequent round, we identify the set of tight nodes $T = \{y \in V : W_y \geq 1/(1+\epsilon) \}$, set $\ell(y) = \ell(y) -1$ for all  $y \in V \setminus T$, and then raise the weights of the edges in the subgraph induced by $V \setminus T$  by a factor of $(1+\epsilon)$. As we only raise the weights of the edges whose both endpoints are slack,  the edge-weights $\{w(e)\}$ continue to be a valid fractional matching in $G$. The algorithm stops when every edge has at least one tight endpoint, so that we are no longer left with any more edges whose weights can be raised.

Clearly, the above algorithm guarantees that the weight of an edge $(x, y) \in E$ is given by $w(x, y) = (1+\epsilon)^{-\ell(x,y)}$. It is also easy to check that the algorithm does not run for more than $L$ rounds, for the following reason. If after starting from $1/n = (1+\epsilon)^{-L}$ we increase the weight of an edge $(x, y)$ more than $L$ times by a factor of $(1+\epsilon)$,  then we would end up having $w(x, y) \geq (1+\epsilon) > 1$, and this would mean that the edge-weights $\{w(e)\}$ no longer form a valid fractional matching. Thus, we conclude that $\ell(y) \in \{0, \ldots, L \}$ for all $y \in V$ at the end of this algorithm. Furthermore, at that time every node $y \in V$ at a nonzero level $\ell(y) > 0$  is tight. The following invariant, therefore, is satisfied. 
\begin{invariant}
\label{inv:mock}
For every node $y$, we have $1/(1+\epsilon) \leq W_y \leq 1$ if $\ell(y) > 0$, and $0 \leq W_y \leq 1$ if $\ell(y) = 0$.
\end{invariant}

Every edge $(x, y) \in E$ has at least one tight endpoint under Invariant~\ref{inv:mock}. To see why this is true, note that if  the edge has some endpoint $z \in \{x, y\}$ at level $\ell(z) > 0$, then Invariant~\ref{inv:mock} implies that the node $z$ is tight. On the other hand, if $\ell(x) = \ell(y) = 0$, then the edge $(x, y)$ has weight $w(x, y) = (1+\epsilon)^{-0} = 1$ and  both its endpoints are tight, for we have $W_x, W_y \geq w(x,y) = 1$.  In other words, the set of tight nodes  constitute a valid vertex cover of the graph $G$. Since every tight node $y$ has weight $1 \geq W_y \geq 1/(1+\epsilon)$, and since every edge $(x,y)$ contributes its own weight $w(x, y)$ towards both $W_x$ and $W_y$, a simple counting argument implies that the number of tight nodes is within a factor $2(1+\epsilon)$ of the sum of the weights of  the edges in $G$. Hence, we have a valid vertex cover and a valid fractional matching whose sizes are within a factor $2(1+\epsilon)$ of each other.  It follows that the set of  tight nodes form  a $2(1+\epsilon)$-approximate minimum vertex cover and that the edge-weights $\{w(e)\}$ form a $2(1+\epsilon)$-approximate maximum fractional matching.

\medskip
\noindent {\bf Making the algorithm dynamic.} In the dynamic setting, all we need to ensure is that we maintain a partition of the node-set $V$ into levels $\{0, \ldots, L\}$ that satisfies Invariant~\ref{inv:mock}. By induction hypothesis, suppose that Invariant~\ref{inv:mock} is satisfied by every node until this point in time. Now, an edge $(u, v)$ is either inserted into or deleted from the graph. The former event  increases the weight $W_x$ of each node $x \in \{u, v\}$ by $(1+\epsilon)^{-\max(\ell(u),\ell(v))}$, whereas the latter event  decreases the weight $W_x$ of each node $x \in \{u, v\}$ by $(1+\epsilon)^{-\max(\ell(u),\ell(v))}$. As a result, one or both of the endpoints $\{u, v \}$ might now violate Invariant~\ref{inv:mock}. For ease of presentation, we say that a node is {\em dirty} if it violates Invariant~\ref{inv:mock}. To be more specific, a node $v$ is {\em dirty} if either (a)  $W_v > 1$ or (b) \{$\ell(v) > 0$ and $W_v < 1/(1+\epsilon)$\}.  In case (a), we say that the node $v$ is {\em up-dirty}, whereas in case (b) we say that the node $v$ is {\em down-dirty}.  To continue with our discussion, we noted that the insertion or deletion of an edge might make one or both of its endpoints dirty.  In such a scenario, we call the subroutine described in Figure~\ref{main:fig:dirty:bhi15}. Intuitively, this subroutine keeps changing the levels of the dirty nodes in a greedy manner till there is no dirty node  (equivalently, till Invariant~\ref{inv:mock} is satisfied). 

In a bit more details, suppose that a node $x$ at level $\ell(x) = i$ is up-dirty. If our goal is to make this node satisfy Invariant~\ref{inv:mock}, then we have to decrease its weight $W_x$. A greedy way to achieve this outcome is to increase its level $\ell(x)$ by one, by setting $\ell(x) = i+1$, without changing the level of any other node. This decreases the weights of all the edges $(x, y) \in E$ incident on $x$ whose other endpoints $y$ lie at levels $\ell(y) \leq i$. The weight of every other edge remains unchanged. Hence, this decreases the weight $W_x$. Note that this step changes the weights of the neighbors $y \in N_x(0,i)$ of $x$ that lie at level $i$ or below. These neighbors, therefore, might now become dirty. Such neighbors will be handled in some future iteration of the {\sc While} loop. Furthermore, it might be the case that the node $x$ itself remains dirty even after this step, since the weight $W_x$ has not decreased by a sufficient amount. In such an event, the node $x$ itself will be handled again in a future iteration of the {\sc While} loop. Next, suppose that the node $x$  is down-dirty. By an analogous argument, we need to increase the weight $W_x$ if we want to make the node $x$ satisfy Invariant~\ref{inv:mock}. Accordingly, we decrease its level $\ell(x)$ in step 5 of Figure~\ref{main:fig:dirty:bhi15}. As in the previous case, this step might lead to some neighbors of $x$ becoming dirty, who will be handled in future iterations of the {\sc While} loop. If the node $x$ itself remains dirty after this step, it will also be handled in some future iteration of the {\sc While} loop. 

To summarize, there is no dirty node when the {\sc While} loop terminates, and hence Invariant~\ref{inv:mock} {\em is} satisfied. But due to the {\em cascading effect} (whereby a given iteration of the {\sc While} loop might create additional dirty nodes), it is not clear why  this simple algorithm will have a small update time. In fact, it is by no means obvious that the {\sc While} loop in Figure~\ref{main:fig:dirty:bhi15} is even guaranteed to terminate. The main result in~\cite{b} was that (a slight variant of) this algorithm actually has an amortized update time of $O(\log n/\epsilon^2)$. Before proceeding any further, however, we ought to highlight the data structures used to implement this algorithm.

\begin{figure}[h!]
	\centerline{\framebox{
			\begin{minipage}{5.5in}
				\begin{tabbing}
					1. \=    {\sc While}  there exists some dirty node $x$: \\
					2.  \= \ \ \ \ \ \ \= {\sc If} the node $x$ is up-dirty, {\sc then} \qquad \qquad \qquad \qquad // We have $W_x > 1$. \\
					3. \> \> \ \ \ \ \ \ \ \= Move it up by one level by setting $\ell(x) \leftarrow \ell(x) + 1$. \\
					4. \> \> {\sc Else if} the node $x$ is down-dirty, {\sc then} \qquad \qquad \ \ \ \ \  // We have $\ell(x) > 0$ and $W_x < 1/(1+\epsilon)$.  \\
					5. \> \> \> Move it down one level by setting $\ell(x) \leftarrow \ell(x) - 1$.
					\end{tabbing}
			\end{minipage}
	}}
	\caption{\label{main:fig:dirty:bhi15} Subroutine:  FIX(.) is called after the insertion/deletion of an edge.}
\end{figure}

\medskip
\noindent {\bf Data structures.} Each node $x \in V$ maintains its  weight $W_x$ and level $\ell(x)$. This  information is sufficient for a node to detect when it becomes dirty. In addition, each node $x \in V$ maintains the following doubly linked lists: For every level $i > \ell(x)$, it maintains the list $E_x(i) = \{ (x, y) \in E : \ell(y) = i \}$ of edges incident on $x$ whose other endpoints lie at level $i$. Thus, every edge $(x, y) \in E_x(i)$ has a weight $w(x,y) = (1+\epsilon)^{-i}$. The node $x$ also maintains the list $E_x^- = \{ (x, y) \in E : \ell(y) \leq \ell(x) \}$ of edges whose other endpoints are at a level that is at most the level of $x$. Thus, every edge $(x, y) \in E_x^-$ has a weight of $w(x, y) = (1+\epsilon)^{-\ell(x)}$. We refer to these lists as the {\em neighborhood lists} of $x$. Intuitively, there is one neighborhood list for each nonempty subset of edges incident on $x$ that have the same weight.  For each edge $(x, y) \in E$, the node $x$ maintains a pointer to its own position in the neighborhood list of $y$ it appears in, and vice versa. Using these pointers, a node can be inserted into or deleted from a neighborhood list in $O(1)$ time. We now bound the time required  to update these data structures during one iteration of the {\sc While} loop.

\begin{claim}
\label{main:cl:ds:up}
Consider a node $x$ that moves from a level $i$ to level $i+1$ during an iteration of the {\sc While} loop in Figure~\ref{main:fig:dirty:bhi15}. Then it takes  $O(|N_x(0, i)|)$ time to update the relevant data structures during that iteration, where $N_x(0, i) = \{ y \in N_x : 0 \leq \ell(y) \leq i \}$ is the set of neighbors of $x$ that lie on or below level $i$. 
\end{claim}

\begin{proof}(Sketch) Consider the event where the node $x$ moves up from level $i$ to level $i+1$. The key observation is this. If the node $x$ has to change its own position in the neighborhood list of another node $y$ due to this event, then we must have $y \in N_x(0,i)$. And as far as changing the neighborhood lists of $x$ itself is concerned, all we need to do is to merge the list $E_x^-$ with the list $E_x(i+1)$, which takes $O(1)$ time.
\end{proof}

\begin{claim}
\label{main:cl:ds:down}
Consider a node $x$ that moves from a level $i$ to level $i-1$ during an iteration of the {\sc While} loop in Figure~\ref{main:fig:dirty:bhi15}. Then it takes  $O(|N_x(0, i)|)$ time to update the relevant data structures during that iteration, where $N_x(0, i) = \{ y \in N_x : 0 \leq \ell(y) \leq i \}$ is the set of neighbors of $x$ that lie on or below level $i$. 
\end{claim}

\begin{proof}(Sketch)
Consider the event where the node $x$ moves down from level $i$ to level $i-1$.  If the node $x$ has to change its own position in the neighborhood list of another node $y$ due to this event, then we must have $y \in N_x(0,i-1)$. On the other hand, in order to  update the neighborhood lists of $x$ itself, we have to visit all the nodes $y \in E_x^-$ one after the other and check their levels. For each such node $y$, if we find that $\ell(y) = i$, then we have to move $y$ from the list $E_x^-$ to the list $E_x(i)$. Thus, the total time spent during this iteration is $O(|N_x(0,i-1)| + |N_x(0,i)|) = O(|N_x(0,1)|)$. The last equality holds since $N_x(0,i-1) \subseteq N_x(0,i)$.
\end{proof}

\subsection{The main technical challenge: Can we bring  down the update time to $O(1)$?}
\label{sec:challenge}

 As we mentioned previously, it was shown in~\cite{b} that the dynamic algorithm described above has an amortized update time of $O(\log n/\epsilon^2)$. In order to prove this bound, the authors in~\cite{b} had to use a complicated potential function. Can we show that (a slight variant of) the same algorithm actually has an update time of $O(1)$ for every fixed $\epsilon$? This seems to be quite a challenging goal, for the following reasons. For now, assume that $\epsilon$ is some small constant.  

In the potential function developed in~\cite{b}, whenever an edge $(u, v)$ is inserted into the graph, we create $O(1) \cdot (L - \max(\ell(u), \ell(v)))$ many {\em tokens}. For each endpoint $x \in \{u, v\}$ and each level $\max(\ell(u), \ell(v)) < i \leq L$, we store $O(1)$ tokens for the node $x$ at level $i$. These tokens are used to account for the time spent on updating the data structures when a node $x$ moves up from a lower level to a higher level, that is, in dealing with up-dirty nodes. It immediately follows that if we only restrict ourselves to the time spent in dealing with up-dirty nodes, then we get an amortized update time of $O(\log n)$. This is because of the following simple accounting: Insertion of an edge $(u, v)$ creates at most $O(L - \max(\ell(u), \ell(v))) = O(\log n)$ many tokens, and each of these tokens is used to pay for one unit of computation performed by our algorithm while dealing with the up-dirty nodes. Next, it is also shown in~\cite{b} that, roughly speaking, over a sufficiently long time horizon the time spent in  dealing with the down-dirty nodes is dominated by the time spent in dealing with the up-dirty nodes. This gives us an overall amortized update time of $O(\log n)$. From this very high level description of the potential function based analysis in~\cite{b}, it seems intrinsically challenging to overcome the $O(\log n)$ barrier. This is because nothing is preventing an edge $(u, v)$ from moving up $\Omega(\log n)$ levels after getting inserted, and according to~\cite{b} the only way we can bound this type of {\em work} performed by the algorithm is by {\em charging} it to the insertion of the edge $(u, v)$ itself. In recent years, attempts were made to overcome this $O(\log n)$ barrier. The papers~\cite{BhattacharyaCH17,Gupta}, for example, managed to improve the amortized update time to $O(1)$, but only at the cost of increasing the approximation ratio from $(2+\epsilon)$ to some unspecified constant $\Theta(1)$. The question of getting $(2+\epsilon)$-approximation in $O(1)$ time, however, remained wide open.

It seems unlikely that we will just  stumble upon a suitable potential function that proves the amortized bound of  $O(1)$ by  trial and error: There are  way too many options to choose from! What we instead need to look for is a systematic {\em meta-method} for finding the suitable potential function -- something that will allow us to prove the optimal possible bound for the given algorithm. This is elaborated upon in Section~\ref{sec:continuous}.

\section{Our technique: A thought experiment with a continuous setting}
\label{sec:continuous}

In order to search for a suitable potential function, we  consider an idealized  setting where the level of a node or an edge  can take {\em any}  (not necessarily integral) value in the {\em continuous} interval $[0, L]$, where $L = \log_{e} n$.  To ease notations, here we assume that the weight of an edge $(x, y)$ is given by $w(x, y) = e^{-\ell(x,y)}$, instead of being equal to $(1+\epsilon)^{-\ell(x,y)}$. This makes it possible to assign each node to a (possibly fractional) level in such a way that the edge-weights $\{w(e)\}$ form a {\em maximal}  fractional matching, and the nodes $y \in V$ with weights $W_y = 1$ form a $2$-approximate minimum vertex cover. We use the notations introduced in the beginning of Section~\ref{sec:old:soda15}. In the idealized setting, we ensure that the following invariant is satisfied.

\begin{invariant}
\label{inv:continuous}
For every node $y \in V$, we have $W_y = 1$ if $\ell(y) > 0$, and $W_y \leq 1$ if $\ell(y) = 0$.
\end{invariant}

\noindent {\bf A static primal-dual algorithm.} As in Section~\ref{sec:old:soda15}, under Invariant~\ref{inv:continuous} the levels of the nodes have a natural primal-dual interpretation. To see this, consider the following static algorithm. We initiate a continuous process  at  time $t = -\log n$. At this stage, we set $w(e) = e^{t} = 1/n$ for every edge $e \in E$. We say that a node $y$ is tight iff $W_y = 1$. Since the maximum degree of a node is at most $n-1$, no node is tight at time $t = -\infty$. With the passage of time, the edge-weights keep increasing exponentially with $t$.   During this process, whenever a node $y$ becomes tight we {\em freeze} (i.e., stop raising) the weights of all its incident edges. The process stops at time $t = 0$. The level of a node $y$ is defined as $\ell(y) = -t_y$, where $t_y$ is the time when it becomes tight during this process. If the node does not become tight till the very end, then  we define $\ell(y) = t_y = 0$. When the process ends at time $t = 0$, it is easy to check that  $w(x, y) = e^{- \ell(x,y)}$ for every edge $(x, y) \in E$ and that Invariant~\ref{inv:continuous} is satisfied.

We claim that under Invariant~\ref{inv:continuous}  the set of tight nodes form a $2$-approximate minimum vertex cover in $G$. To see why this is true, suppose that there is an edge $(x, y)$ between two  nodes $x$ and $y$ with $W_x, W_y < 1$. According to Invariant~\ref{inv:continuous}, both the  nodes $x, y$ are at level $0$. But this implies that $w(x,y) = e^{-0} = 1$ and hence $W_x, W_y \geq w(x,y) = 1$, which leads to a contradiction.  Thus, 
 the set of tight nodes must be a vertex cover in $G$. Since $W_v = 1$ for every tight node $v \in V$, and since every edge $(u, v) \in E$ contributes the weight $w(x, y)$ to both $W_x$ and $W_y$, a simple counting argument implies that the edge-weights $\{w(e) \}$ forms a fractional matching in $G$ whose size is at least $(1/2)$ times the number of tight nodes. The claim now follows from the duality between maximum fractional matching and minimum vertex cover. 

\medskip
We will now describe how Invariant~\ref{inv:continuous} can be maintained in a dynamic setting -- when edges are getting inserted into or deleted from the graph. For ease of exposition, we will use Assumption~\ref{assume:different}. Since the level of a node can take {\em any} value in the {\em continuous} interval $[0, L]$, this  does not appear to be too restrictive.

\begin{assumption}
\label{assume:different}
For any two nodes $x \neq y$, if $\ell(y), \ell(y) > 0$, then we have $\ell(x) \neq \ell(y)$. 
\end{assumption}

\noindent 
{\bf Notations.} Let $N^+_x = \{ y \in V : (x, y) \in E \text{ and } \ell(y) > \ell(x) \}$ denote the set of {\em up-neighbors} of a node $x \in V$, and let $N^-_x = \{ y \in V : (x, y) \in E \text{ and } \ell(y) < \ell(x) \}$ denote the set of {\em down-neighbors} of $x$.  Assumption~\ref{assume:different} implies that $N_x = N^+_x \cup N^-_x$. Finally, let $W_x^+$ (resp. $W_x^-$) denote the total weight of the edges incident on $x$ whose other endpoints are in $N_x^+$ (resp. $N_x^-$). We thus have $W_x = W_x^+ + W_x^-$ and $W_x^- = |N_x^-| \cdot e^{-\ell(x)}$. We will use these notations throughout the rest of this section.

\medskip
\noindent {\bf Insertion or deletion of an edge $(u, v)$.}  We focus only on the case of an edge insertion, as the case of edge deletion can be handled in an analogous manner. Consider the {\em event} where an edge $(u, v)$ is inserted into the graph. By induction hypothesis Invariant~\ref{inv:continuous} is satisfied just before this event, and without any loss of generality suppose that $\ell(v) = i \geq \ell(u) = j$ at that time. For ease of exposition, we assume that $i, j > 0$: the other case can be dealt with using similar ideas. Then, we have $W_u = W_v = 1$ just before the event (by Invariant~\ref{inv:continuous}) and $W_u = W_v = 1+e^{-i}$ just after the event (since $i \geq j$). So the nodes $u$ and $v$ violate Invariant~\ref{inv:continuous} just after the event. We now explain the process by which the nodes change their levels so as to ensure that Invariant~\ref{inv:continuous} becomes satisfied again. This process consists of two {\em phases} -- one for each endpoint. $x \in \{u, v\}$.  We now describe each of these phases.

\medskip
\noindent {\bf Phase I:} {\em This phase is defined by a continuous process which is driven by the node $v$.} Specifically, in this phase the node $v$ continuously increases its level so as to decrease its weight $W_v$. The process stops when the weight $W_v$ becomes equal to $1$.  During the same process, every other node $x \neq v$ continuously changes its level so as to ensure that its weight $W_x$ remains fixed.\footnote{To be precise, this statement does not apply to the nodes at level $0$. A node $x$ with $\ell(x) = 0$ remains at level $0$ as long as $W_x < 1$, and starts moving upward only when its weight $W_x$ is about to exceed $1$. But, morally speaking, this  does not add any new perspective to our discussion, and henceforth we ignore this case.} This creates a cascading effect which leads to a long chain of interdependent movements of nodes. To see this, consider an infinitesimal time-interval  $[t, t+dt]$ during which the node $v$ increases its level from $\ell(v)$ to $\ell(v) + d \ell(v)$. The weight of every edge $(v, x) \in E$ with $x \in N_v^-$ {\em decreases} during this interval, whereas the weight of every other edge remains unchanged. Thus, during this interval, the upward movement of the node $v$ leads to a {\em decrease} in the weight $W_x$ of every neighbor $x \in N_v^-$. Each such node $x \in N_v^-$ wants to {\em nullify} this effect and ensure that $W_x$ remains fixed. Accordingly,  each such node $x \in N_v^-$ {\em decreases} its level during the same infinitesimal time-interval $[t, t+dt]$ from $\ell(x)$ to $\ell(x) + d \ell(x)$, where $d \ell(x) < 0$. The value of $d \ell(x)$ is such that $W_x$ actually remains unchanged during the time-interval $[t, t+dt]$. Now,  the weights of the neighbors $y \in N_x^-$ of $x$ also get affected as $x$ changes its level, and as a result each such node $y$  also changes its level so as to preserve its own weight $W_y$, and so on and so forth. We emphasize that all these movements of different nodes occur {\em simultaneously}, and in a continuous fashion.  Intuitively, the set of nodes form a  self-adjusting system --  akin to a spring. Each node moves in a way which ensure that its weight becomes (or, remains equal to) a ``critical value". For the node $u$ this critical value is $1+e^{-i}$, and for every other node (at a nonzero level) this critical value is equal to $1$. Thus, every node other than $u$ satisfies Invariant~\ref{inv:continuous} when Phase I ends. At this point, we initiate Phase II described below.

\medskip
\noindent {\bf Phase II:} {\em This phase is defined by a continuous process which is driven by the node $u$.} Specifically, in this phase the node $u$ continuously increases its level so as to decrease its weight $W_u$. The process stops when $W_u$ becomes equal to $1$. As in Phase I, during the same process every other node $x \neq u$ continuously changes its level so as to ensure that $W_x$ remains fixed. Clearly, Invariant~\ref{inv:continuous} is satisfied when Phase II ends.

\medskip
\noindent {\bf ``Work": A proxy for update time.} We cannot implement the above continuous process using any data structure, and hence we cannot meaningfully talk about the {\em update time} of our algorithm in the idealized, continuous setting. To address this issue, we introduce the notion of {\em work}, which is defined as follows. We say that our algorithm performs $\delta \geq 0$ work  whenever it changes the level $\ell(x, y)$ of an edge $(x, y)$ by $\delta$.  Note that $\delta$ can take any arbitrary fractional value. To see how the notion of work relates to the notion of update time from Section~\ref{sec:old:soda15}, recall Claim~\ref{main:cl:ds:up} and Claim~\ref{main:cl:ds:down}. They state that whenever a node $x$ at level $\ell(x) = k$ moves up or down one level, it takes $O(|N_x(0,k)|)$ time to update the relevant data structures. A moment's thought will reveal that in the former case (when the node moves up) the total work done is equal to $|N_x(0, k)|$, and in the latter case (when the node moves down) the work done is equal to $|N_x(0,k-1)|$. Since $N_x(0,k-1) \subseteq N_x(0,k)$, we have  $|N_x(0,k-1)| \leq |N_x(0,k)|$. Thus, the work done by the algorithm is upper bounded by (and, closely related to) the time spent to update the date structures.  In light of this observation, we now focus on analyzing the work done by our algorithm in the continuous setting.

\subsection{Work done in handling the insertion or deletion of an edge $(u, v)$}
\label{sec:thought:exp} 

We focus on the case of an edge-insertion. The case of an edge-deletion can be analyzed using similar ideas. Accordingly, suppose that an edge $(u, v)$, where $\ell(v) \geq \ell(u)$, gets inserted into the graph. We first analyze the work done in Phase I, which is driven by the movement of $v$. Without any loss of generality, we assume that $v$ is changing its level in such a way that its weight $W_v$ is decreasing at unit-rate. Every other node $x$ at a nonzero level wants to preserve its weight $W_x$ at its current value. Thus, we have:
\begin{eqnarray}
\label{eq:lambda:rate:v}
\frac{d W_v}{dt} & = & -1 \\ 
\label{eq:lambda:rate:x} 
\frac{d W_x}{dt} & = & 0 \text{ for all nodes } x \neq v \text{ with } \ell(x) > 0.
\end{eqnarray} 
\noindent {\bf A note on how the sets $N^-_x$ and $N_x^+$ and the weights $W_x^-$ and $W_x^+$ change with time:} We will soon write down a few differential equations, which capture the behavior of the continuous process unfolding in Phase I  during an infinitesimally small time-interval $[t, t+dt]$. Before embarking on this task, however, we need to clarify the following important issue. Under Assumption~\ref{assume:different}, at time $t$ the (nonzero) levels of the nodes take distinct, finite values. Thus,  we have $\ell_t(x) \neq \ell_t(y)$ for any two nodes $x \neq y$ with $\ell_t(x), \ell_t(y) > 0$, where $\ell_t(z)$ denotes the level of a node $z$ at time $t$. The level of a node can only change by an infinitesimally small amount during the time-interval $[t+dt]$. This implies that if $\ell_t(x) > \ell_t(y)$ for any two nodes $x, y$, then we also have $\ell_{t+dt}(x) > \ell_{t+dt}(y)$. In words, while writing down a differential equation we can assume that the sets $N_x^+$ and $N_x^-$ remain unchanged throughout the infinitesimally small time-interval $[t, t+dt]$.\footnote{The sets $N_x^+$ and $N_x^-$ will indeed change over a sufficiently long, {\em finite} time-interval. The key observation is that we can ignore this change while writing down a differential equation for an infinitesimally small time-interval.} But, this observation does not apply to the weights $W_x^-$ and $W_x^+$, for  the weight $w(x, y) = e^{-\max(\ell(x), \ell(y))}$  of an edge will change if we move the level of its higher endpoint by an infinitesimally small amount.

\medskip
Let $s_x = \frac{d \ell(x)}{dt}$ denote the {\em speed} of a node $x \in V$. It is the rate at which the node $x$ is changing its level. Let $f(x, x')$ denote the rate at which the weight $w(x, x')$ of an edge $(x, x') \in E$ is changing. Note that:
\begin{equation}
\label{eq:lambda:3}
\text{If } \ell(x) > \ell(x'), \text{ then } f(x, x')  = \frac{d w(x, x')}{dt} = \frac{d e^{-\ell(x)}}{dt} = \frac{d \ell(x)}{dt} \cdot \frac{d e^{-\ell(x)}}{d\ell(x)}= - s_{x} \cdot w(x, x').
\end{equation}

\noindent Consider a node $x \neq v$ with $\ell(x) > 0$. By~(\ref{eq:lambda:rate:x}), we have $\frac{d W_x}{dt} = 0$. Hence, we derive that:
$$0 = \frac{d W_x}{dt} = \sum_{x' \in N_x^-} \frac{d w(x, x')}{dt} + \sum_{x' \in N_x^+} \frac{d w(x, x')}{dt} =  \sum_{x' \in N_x^-} f(x,x') + \sum_{x' \in N_x^+} f(x,x').$$
Rearranging the terms in the above equality, we get:
\begin{equation}
\label{eq:lambda:4}
\sum_{x' \in N^-_x} f(x, x') = - \sum_{x' \in N^+_x} f(x,x') \text{ for every node } x \in V \setminus \{v\} \text{ with } \ell(x) > 0.
\end{equation}
Now, consider the node $v$. By~(\ref{eq:lambda:rate:v}), we have $\frac{d W_v}{dt} = \sum_{x' \in N^-_v} \frac{d w(x', v)}{dt} = -1$. Hence, we get:
\begin{equation}
\label{eq:lambda:v}
\sum_{x' \in N^-_v} f(x', v)  = \frac{d W_v}{dt} = -1.
\end{equation}
Conditions~(\ref{eq:lambda:4}) and~(\ref{eq:lambda:v}) are reminiscent of a flow constraint. Indeed, the entire process can be visualized as follows. Let $|f(x, y)|$ be the {\em flow} passing through an edge $(x, y) \in E$. We {\em pump}  $1$ unit of flow {\em into} the node $v$ (follows from~(\ref{eq:lambda:rate:v})). This flow then splits up evenly among all the edges $(x, v) \in E$ with $x \in N_v^-$ (follows from~(\ref{eq:lambda:v}) and~(\ref{eq:lambda:3})). As we sweep across the system down to lower and lower levels, we see the same phenomenon: The flow coming into a node $x$ from its neighbors  $y \in N_x^+$ splits up evenly among its neighbors $y \in N_x^-$ (follows from~(\ref{eq:lambda:4}) and~(\ref{eq:lambda:3})). 
Our goal is to analyze the work done by our algorithm. Towards this end, let $P_{(x, x')}$ denote the {\em power} of an edge $(x, x') \in E$. This is the amount of  work being done by the algorithm on the edge $(x, x')$ per time unit. Thus, from~(\ref{eq:lambda:3}), we get:
\begin{equation}
\label{eq:lambda:5} 
\text{If } \ell(x) > \ell(x'), \text{ then } P_{(x, x')} = |s_x| = \frac{|f(x, x')|}{w(x, x')} = |f(x, x')| \cdot e^{\ell(x)}.
\end{equation}
Let $P_x$ denote the {\em power} of a node $x \in V$. We define it to be the amount of work being done by the algorithm for changing the level of  $x$ per time unit. This is the sum of the powers of the edges whose levels change due to the node $x$ changing its own level. Let $f^-(x) = \sum_{x' \in N_x^-} f(x, x')$.  From~(\ref{eq:lambda:3}), it follows that either $f(x', x) \geq 0$ for all $x' \in N_x^-$, or $f(x',x) \leq 0$ for all $x' \in N_x^-$. In other words, every term in the sum $\sum_{x' \in N_x^-} f(x, x')$ has the same sign. So the quantity $|f^-(x)|$ denotes the total flow moving from the node $x$ to its neighbors $x' \in N_x^-$, and we derive that:
\begin{equation}
\label{eq:lambda:6}
P_x = \sum_{x' \in N^-_x} P_{(x, x')} = e^{\ell(x)} \cdot \sum_{x' \in N_x^-} \left| f(x, x') \right| = e^{\ell(x)} \cdot \left| f^-(x) \right|
\end{equation}
The total work done by the algorithm per time unit in Phase I is equal to $\sum_{(x, x') \in E} P_{(x, x')} = \sum_{x \in V} P_x$. We will like to upper bound this sum.
We now make the following important observations. First, since the flow only moves downward after getting pumped into the node $v$ at unit rate,  conditions~(\ref{eq:lambda:4}) and~(\ref{eq:lambda:v}) imply that:
\begin{equation}
\label{eq:flow}
\sum_{x : \ell(x) = k} \left| f^-(x) \right| \leq 1 \text{ at every level }  k \leq \ell(v).
\end{equation}
Now, suppose that we get extremely lucky, and we end up in a situation where the levels of all the nodes are integers (this was the case in Section~\ref{sec:old:soda15}). In this situation, as the flow moves down the system to lower and lower levels, the powers of the  nodes  decrease geometrically as per condition~(\ref{eq:lambda:6}). Hence, applying~(\ref{eq:flow}) we can upper bound the sum $\sum_{x} P_x$ by the geometric series $\sum_{k = 0}^{\ell(v)} e^{k}$. This holds since:
\begin{equation}
\label{eq:flow:2}
\sum_{x \in V : \ell(x) = k} P_x = \sum_{x \in V : \ell(x) = k} |f^-(x)| \cdot e^k \leq e^k \text{ at every level } k \leq \ell(v).
\end{equation}
Thus, in Phase I the algorithm performs $\sum_{k = 0}^{\ell(v)} e^k = O\left(e^{\ell(v)}\right)$ units of work per time unit. Recall that Phase I was initiated after the insertion of the edge $(u, v)$, which increased the weight $W_v$ by (say) $\eta_v$.  During this phase the node $v$ decreases its weight $W_v$ at unit rate, and the process stops when $W_v$ becomes equal to $1$. Thus, from the discussion so far we expect Phase I to last for $\eta_v$ time-units. Accordingly, we also expect the total work done by the algorithm in Phase I  to be at most $O\left(e^{\ell(v)}\right) \cdot \eta_v$. Since $\eta_v = e^{-\max(\ell(u), \ell(v))} \leq e^{-\ell(v)}$, we expect that the algorithm will do at most $O\left(e^{\ell(v)}\right) \cdot e^{-\ell(v)} = O(1)$ units of work in Phase I. A similar reasoning applies for Phase II as well. This gives an intuitive explanation as to why an appropriately chosen variant of the BHI15 algorithm~\cite{b} should have $O(1)$ update time for every constant $\epsilon > 0$.

\subsection{Towards analyzing the ``real-world", discretized setting}

Towards the end of Section~\ref{sec:thought:exp} we made a crucial assumption, namely, that the levels of the nodes are integers. It turns out that if we want to enforce this condition, then we can no longer maintain an {\em exact} maximal fractional matching and get an approximation ratio of $2$. Instead, we will have to satisfied with a fractional matching that is approximately maximal, and the corresponding vertex cover we get will be a $(2+\epsilon)$-approximate minimum vertex cover. Furthermore,  in the idealized continuous setting we could get away with moving the level of a node $x$, whose weight $W_x$ has only slightly deviated from $1$, by any arbitrarily small amount and thereby doing arbitrarily small amount of work on the node at any given time. This is why the intuition we got out of the above discussion also suggests that the overall update time should be $O(1)$ in the {\em worst case}. This will no longer be possible in the {\em real-world}, where the levels of the nodes need to be integers. In the real-world, a node $x$ can move to a different integral level only after its weight $W_x$ has changed by a sufficiently large amount, and the work done to move the node to a different level can be quite significant. This is why our analysis in the discretized, real-world  gets an {\em amortized} (instead of worst-case) upper bound of $O(\epsilon^{-2})$ on the update time of the algorithm.

Coming back to the continuous world, suppose that we pump in an infinitesimally small  $\delta$ amount of weight into a node $x$ at a level $\ell(x) = k > 0$ at unit rate. The process, therefore, lasts for  $\delta$ time units. During this process, the level of the node $x$  increases by an infinitesimally small amount so as to ensure that its weight $W_x$ remains equal to $1$. The work done per time unit on the node $x$ is equal to $P_x$. Hence, the total work done on the node $x$ during this event is given by: 
$$\delta \cdot P_x = \delta \cdot \sum_{y \in N^-_x} P_y = \delta \cdot \sum_{y \in N^-_x} |f(x, y)| \cdot e^{\ell(x)} = \delta \cdot e^{\ell(x)} \cdot \left|\sum_{y \in N^-_x} f(x,y)\right| = \delta \cdot e^{\ell(x)}.$$
In this  derivation, the first two steps follow from~(\ref{eq:lambda:5}) and~(\ref{eq:lambda:6}). The third step holds since $f(x, y) < 0$ for all $y \in N^-_x$, as the node  $x$   moves up to a higher level. The fourth step follows from~(\ref{eq:lambda:v}). Thus, we note that:
\begin{observation}
\label{ob:main}
To change the weight $W_x$ by $\delta$, we need to perform  $\delta \cdot e^{\ell(x)}$ units of work  on the node $x$. 
\end{observation}
The intuition derived from Observation~\ref{ob:main} will guide us while we design a potential function for bounding  the amortized update time in the  ``real-world". This is shown in Section~\ref{sec:algorithm:new}. 

\section{An overview of our algorithm and the analysis in the ``real-world"}
\label{sec:algorithm:new}

To keep the presentation as modular as possible, we describe the algorithm itself in Section~\ref{sec:describe}, which happens to be almost the same as the BHI15 algorithm from Section~\ref{sec:old:soda15}, with one crucial twist. Accordingly, our main goal in Section~\ref{sec:describe} is to point out the  difference between the new algorithm and the old one. We also explain why this difference does not impact in any significant manner the approximation ratio of $(2+\epsilon)$ derived in Section~\ref{sec:old:soda15}. Moving forward, in Section~\ref{sec:analyze} we present a very high level overview of our new potential function based analysis of the algorithm from Section~\ref{sec:describe}, which gives the desired  bound of $O(1/\epsilon^2)$ on the amortized update time. See Part~\ref{part:full} for the full version of the algorithm and its analysis.

\subsection{The algorithm}
\label{sec:describe}

We start by setting up the necessary notations.  We use all the notations introduced in the beginning of Section~\ref{sec:old:soda15}. In addition, for every node $x \in V$ and every level $0 \leq i \leq L$, we let $W_{x \rightarrow i} = \sum_{y \in N_x} (1+\epsilon)^{-\max(\ell(y), i)}$ denote what  the weight of  $x$ would have been if we were to place  $x$ at level $i$, without changing the level of any other node. Note that $W_{x \rightarrow i}$ is a monotonically (weakly) decreasing function of $i$, for the following reason. As we increase the level of $x$ (say) from $i$ to $(i+1)$, all its incident edges $(x, y) \in E$ with $y \in N_x(0, i)$ decrease their weights, and the weights of all its other incident edges  remain unchanged. 

\medskip
\noindent {\bf Up-dirty and down-dirty nodes.} We use the same definition of a {\em down-dirty} node as in Section~\ref{sec:old:soda15} (see  the second paragraph after Invariant~\ref{inv:mock}) -- a node $x$ is {\em down-dirty} iff $\ell(x) > 0$ and $W_x < 1/(1+\epsilon)$.  But we slightly change the definition of an {\em up-dirty} node. Specifically, here we say that a node $x \in V$  is {\em up-dirty} iff $W_x > 1$ {\em and} $W_{x \rightarrow \ell(x)+1} > 1$. As before, we say that a node is {\em dirty} iff it is either up-dirty or down-dirty. 

\medskip
\noindent {\bf Handling the insertion or deletion of an edge.} The pseudocode for handling the insertion or deletion of an edge $(u, v)$ remains the same as in Figure~\ref{main:fig:dirty:bhi15} -- although  the conditions which specify when a node is up-dirty have changed. As far as the time spent in implementing the subroutine in Figure~\ref{main:fig:dirty:bhi15} is concerned, it is not difficult to come up with suitable data structures so that Claim~\ref{main:cl:ds:up} and Claim~\ref{main:cl:ds:down} continue to hold. 

\medskip
\noindent {\bf Approximation ratio.} Clearly, this new algorithm ensures that there is no dirty node when it is done with handling the insertion or deletion of an edge. We can no longer claim, however, that Invariant~\ref{inv:mock} is satisfied. This is because we have changed the definition of an up-dirty node. To address this issue, we make the following key observation: If a node $x$ with $W_x > 1$ is {\em not} up-dirty according to the new definition, then we must have $W_x \leq (1+\epsilon)$. To see why this true, suppose that we have a node $x$ with $W_x > (1+\epsilon)$ that is {\em not} up-dirty. If we move this node up by one level, then every edge incident on $x$ will decrease its weight by at most a factor of $(1+\epsilon)$, and hence the weight $W_x$ will also decrease by a factor of at most $(1+\epsilon)$. Therefore, we infer that $W_{x \rightarrow \ell(x)+1} \geq W_x/(1+\epsilon) > 1$, and the node $x$ {\em is} in fact up-dirty. This leads to a contradiction. Hence, it must be the case that if a node $x$ with $W_x > 1$ is up-dirty, then $W_x \leq (1+\epsilon)$. This observation implies that if there is no dirty node, then the following conditions are satisfied. (1) $W_x \leq (1+\epsilon)$ for all nodes $x \in V$. (2) $W_x \geq 1/(1+\epsilon)$ for all nodes $x \in V$ at levels $\ell(x) > 0$. Accordingly, we get a valid fractional matching if we scale down the edge-weights by factor of $(1+\epsilon)$. As before, the set of nodes $x$ with $W_x \geq 1/(1+\epsilon)$ forms a valid vertex cover. A simple counting argument (see the paragraph after Invariant~\ref{inv:mock}) implies that the size of this fractional matching is within a  $2(1+\epsilon)^2$ factor of the size of this vertex cover. Hence, we get an approximation ratio of $2(1+\epsilon)^2$. Basically, the approximation ratio degrades only by a factor of $(1+\epsilon)$ compared to the analysis in Section~\ref{sec:old:soda15}.

\subsection{Bounding the amortized update time}
\label{sec:analyze}

Our first task is to find a discrete, real-world analogue of Observation~\ref{ob:main} (which holds only in the continuous setting). This is done in Claim~\ref{cl:up:brand:new} below. This relates the time required to move up a node $x$ from level $k$ to level $k+1$ with the change in its weight $W_x$ due to the same event.

\begin{claim}
\label{cl:up:brand:new}
Suppose that a node $x$ is moving up from level $k$ to level $k+1$ during an iteration of the {\sc While} loop in Figure~\ref{main:fig:dirty:bhi15}. Then it takes $O\left((W_{x \rightarrow k} - W_{x\rightarrow k+1}) \cdot \epsilon^{-1} \cdot (1+\epsilon)^k \right)$  time to update the relevant data structures during this iteration.
\end{claim}

\begin{proof}
As the node $x$ moves up from level $k$ to level $k+1$, the weight of every edge $(x, y) \in E$ with $y \in N_x(0,k)$ decreases from $(1+\epsilon)^{-k}$ to $(1+\epsilon)^{-(k+1)}$, whereas the weight of every other edge remains unchanged. Hence, it follows that: 
$$W_{x \rightarrow k} - W_{x \rightarrow k+1} = \left| N_x(0,k) \right| \cdot \left( (1+\epsilon)^{-k} - (1+\epsilon)^{-(k+1)}\right) =  \left| N_x(0,k) \right| \cdot \epsilon \cdot (1+\epsilon)^{-(k+1)}.$$
Rearranging the terms in the above equality, we get:
$$\left| N_x(0,k) \right| = \left(W_{x \rightarrow k} - W_{x \rightarrow k+1}\right) \cdot \epsilon^{-1} \cdot (1+\epsilon)^{k+1} = O\left((W_{x \rightarrow k} - W_{x \rightarrow k+1}) \cdot \epsilon^{-1} \cdot (1+\epsilon)^{k}\right).$$
The desired proof now follows from Claim~\ref{main:cl:ds:up}.
\end{proof}

Next, consider a node $x$ that is moving {\em down} from level $k$ to level $k-1$. We use a different accounting scheme to bound the time spent during this event. This is because the work done on the node $x$ during such an event is equal to $|N_x(0,k-1)|$, but it takes $O\left(| N_x(0,k)|\right)$ time to update the relevant data structures (see Claim~\ref{main:cl:ds:down} and the last paragraph before Section~\ref{sec:thought:exp}). Note that $N_x(0,k-1) \subseteq N_x(0,k)$. Thus, although it is possible to bound the work done during this event in a manner analogous to Claim~\ref{cl:up:brand:new}, the bound obtained in that manner might be significantly less than the actual time spent in updating the data structures during this event. Instead, we bound the time spent during this event as specified in Claim~\ref{cl:up:brand:new} below.

\begin{claim}
\label{cl:down:brand:new}
Consider a node $x$  moving down from level $k$ to level $k-1$ during an iteration of the {\sc While} loop in Figure~\ref{main:fig:dirty:bhi15}. Then it takes $O\left((1+\epsilon)^{k-1} \right)$  time to update the relevant data structures during this iteration.
\end{claim}

\begin{proof}
By Claim~\ref{main:cl:ds:down}, it takes $O\left(|N_x(0,k)|\right)$ time to update the relevant data structures when the node $x$ moves down from level $k$ to level $k-1$. We will now show that $|N_x(0,k)| = O\left( (1+\epsilon)^k\right)$. To see why this is true, first note that the node $x$ moves down from level $k$ only if it is down-dirty at that level (see step 4 in Figure~\ref{main:fig:dirty:bhi15}). Hence, we get: $W_{x \rightarrow k} < (1+\epsilon)^{-1}$. When the node is at level $k$, every edge $(x, y) \in E$ with $y \in N_x(0,k)$ has a weight $w(x, y) = (1+\epsilon)^{-k}$. It follows that $(1+\epsilon)^{-k} \cdot |N_x(0,k)| \leq W_{x \rightarrow k} < (1+\epsilon)^{-1}$. Rearranging the terms in the resulting inequality, we get: $|N_x(0,k)| < (1+\epsilon)^{k-1}$.
\end{proof}

\noindent {\bf Node potentials and energy.}  In order to bound the amortized update time, we introduce the notions of potentials and energy of nodes. Every node $x \in V$ stores nonnegative potentials $\Phi^{\uparrow}(x, k)$, $\Phi^{\downarrow}(x, k)$ and energies $\E^{\uparrow}(x, k)$, $\E^{\downarrow}(x,k)$ at every level $0 \leq k \leq L$. The potential $\Phi^{\uparrow}(x, k)$ and the energy $\E^{\uparrow}(x, k)$ are used to account for the time spent in moving the node $x$ {\em up} from level $k$ to level $k+1$. Similarly, the potential $\Phi^{\downarrow}(x, k)$ and the energy $\E^{\downarrow}(x, k)$ are used to account for the time spent in moving the node $x$ {\em down} from level $k$ to level $k-1$.  Each unit of potential at level $k$  results in $(1+\epsilon)^k \cdot \epsilon^{-1}$ units of energy. Accordingly, we refer to this quantity $(1+\epsilon)^k \cdot \epsilon^{-1}$ as the {\em conversion rate} between potential and energy at level $k$. 
\begin{equation}
\label{main:eq:energy:new}
\text{We have } \E^{\gamma}(x, k) = \Phi^{\gamma}(x, k) \cdot (1+\epsilon)^k \cdot \epsilon^{-1} \text{ for all } x \in V, \text{ all }  k \in [0, L], \text{ and all } \gamma \in \{ \uparrow, \downarrow\}. 
 \end{equation}
The  potentials stored by a node $x$ across different levels depends on its weight $W_x$. To be more specific, it depends on whether $W_x > 1$ or $W_x \leq 1$. We first define the potentials  stored by a node $x$ with weight $W_x > 1$. Throughout the following discussion, we crucially rely upon the fact that $W_{x \rightarrow k}$ is a monotonically (weakly) decreasing function of $k$.  For any node $x \in V$ with $W_x > 1$, let $\ell^{\uparrow}(x)$  be the maximum level $k \in \{\ell(x), \ldots, L\}$ where $W_{x \rightarrow k} \geq 1$. The potentials $\Phi^{\uparrow}(x, k), \Phi^{\downarrow}(x,k)$  are then defined as follows.
\begin{eqnarray}
\label{main:eq:potential:up:new}
\text{If a node } x  \text{ has  } W_x > 1, \text{ then }   \Phi^{\uparrow}(x, k) & = & \begin{cases} 0 & \text{ for all } \ell^{\uparrow}(x) < k \leq L; \\
W_{x \rightarrow k} - 1 & \text{ for } k = \ell^{\uparrow}(x); \\
W_{x \rightarrow k} - W_{x \rightarrow k+1} & \text{ for all }  \ell(x) \leq  k < \ell^{\uparrow}(x); \\
0 & \text{ for all } 0 \leq k < \ell(x).
\end{cases}  \\ 
\Phi^{\downarrow}(x, k) & = & 0 \qquad  \text{ for all } 0 \leq k \leq L. \label{main:eq:potential:up:new:1}
\end{eqnarray}
It is easy to check that under~(\ref{main:eq:potential:up:new}) we have $\Phi^{\uparrow}(x, k) \geq 0$ at every level $k$. Summing over all the levels, the total potential associated with this node $x$ is given by: $\sum_{k=1}^L \Phi^{\uparrow}(x, k) = W_{x \rightarrow \ell(x)} - 1 = W_x - 1$. 

We now give some intuitions behind~(\ref{main:eq:potential:up:new}) and~(\ref{main:eq:potential:up:new:1}). Although these equations might seem daunting at first glance, they in fact follow quite  naturally  from Claim~\ref{cl:up:brand:new}. To see this, first note that as long as $W_x > 1$, the node $x$ never has to decrease its level by moving downward. Hence, if $W_x > 1$, then it is natural to set $\Phi^{\downarrow}(x, k) = 0$ at evert level $k$. Next, consider the ``interesting" scenario when the node $x$ is moving up from its current level $\ell(x) = i$ to level $i+1$. According to step 2 in Figure~\ref{main:fig:dirty:bhi15}, this means that the node $x$ is up-dirty at level $i$. From the new definition of a up-dirty node introduced in this section, it follows that $W_{x \rightarrow i} \geq W_{x \rightarrow i+1} > 1$. Just {\em before} the node $x$ moves up from level $i$, we have $\Phi^{\uparrow}(x, i) = W_{x \rightarrow i} - W_{x \rightarrow i+1}$, and just {\em after} the node $x$ moves to level $i+1$ we have $\Phi^{\uparrow}(x, i)  = 0$. Accordingly, we say that the node $x$ {\em releases} $(W_{x \rightarrow i} - W_{x \rightarrow i+1})$ units of potential at level $i$ during this event. As per our conversion ratio between potential and energy defined in~(\ref{main:eq:energy:new}), the node $x$ also releases  $(W_{x \rightarrow i} - W_{x \rightarrow i+1}) \cdot (1+\epsilon)^i \cdot \epsilon^{-1}$ units of energy during this event. Claim~\ref{cl:up:brand:new} now implies that the time spent in updating the data structures during this event is at most the energy released by the node $x$ during the same event. Note that this event does not affect the potentials of the node $x$ at any other level $j \neq i$.

Below, we  define the potentials  of a  node $x$ with weight $W_x \leq 1$.
\begin{eqnarray}
\label{main:eq:potential:down:new}
\text{If a node } x  \text{ has } W_x \leq 1, \text{ then }   \Phi^{\downarrow}(x, k) & = & \begin{cases}  0 & \text{ for all } \ell(x) < k \leq L; \\
 1 - W_{x \rightarrow k}  & \text{ for }  k = \ell(x); \\
0 & \text{ for all } 0 \leq k  < \ell(x).
\end{cases} \\
\Phi^{\uparrow}(x, k) & = & 0 \qquad \text{ for all } 0 \leq k \leq L. \label{main:eq:potential:down:new:1}
\end{eqnarray}
As one might expect, the above equations should be seen as naturally following from Claim~\ref{cl:down:brand:new}. To see this, first note that the node $x$ does not need to move up from its current level as long as $W_x \leq 1$. Hence, if $W_x \leq 1$, then it makes sense to define $\Phi^{\uparrow}(x, k) = 0$ at every level $0 \leq k \leq L$. Next, consider the event where the node $x$ is moving down from level $i$ to level $i-1$. Then step 4 in Figure~\ref{main:fig:dirty:bhi15} ensures that the node $x$ is down-dirty at level $i$, so that $W_{x \rightarrow i} < (1+\epsilon)^{-1}$. Thus, we have $\Phi^{\downarrow}(x, i) = 1 - W_{x \rightarrow i} > \epsilon \cdot (1+\epsilon)^{-1}$ just {\em before} the event during which the node $x$ moves down from level $i$ to level $i-1$, whereas we have $\Phi^{\downarrow}(x, i) = 0$ just {\em after} the same event. Accordingly, we say that the node $x$ {\em releases} at least $\epsilon \cdot (1+\epsilon)^{-1}$ units of potential at level $k$ during this event. As per the conversion ratio between potential and energy defined in~(\ref{main:eq:energy:new}), the node $x$ also releases at least $\epsilon \cdot (1+\epsilon)^{-1} \cdot (1+\epsilon)^i \cdot \epsilon^{-1} = (1+\epsilon)^{i-1}$ units of energy at level $i$ during this event. On the other hand, Claim~\ref{cl:down:brand:new} states that the time spent in updating the data structures during this event is at most $O\left((1+\epsilon)^{i-1}\right)$.  So the time spent during this event is at most the energy released by the node $x$ at level $i$. However, in contrast with the discussion following~(\ref{main:eq:potential:up:new}) and~(\ref{main:eq:potential:up:new:1}), here $1-W_{x \rightarrow i-1}$ units of potential and $(1- W_{x \rightarrow i-1}) \cdot (1+\epsilon)^{i-1} \cdot \epsilon^{-1}$ are {\em created} when the node $x$ moves down to level $i-1$. The energy released by the node at level $i$ only accounts for the time spent in updating the data structures. We need to delve into a deeper analysis of the entire framework in order to bound this new energy that gets created  {\em as a result of} moving the node $x$ down to a lower level, without which our proof for the bound on the amortized update time will remain incomplete.  Before embarking on this task, however, we  formally clarify the way we are going to use  two phrases: {\em potential (resp. energy) absorbed by a node}, and {\em potential (resp. energy) released by a node}. This is explained below.

\medskip

Fix a node $x \in V$ and a level $k \in[0, L]$. Consider an event which (possibly) changes the potentials $\Phi^{\uparrow}(x, k)$ and $\Phi^{\downarrow}(x, k)$. Let $\Phi^{\uparrow}_0(x, k)$ and $\Phi^{\uparrow}_1(x, k)$ respectively denote the value of $\Phi^{\uparrow}(x, k)$ before and after this event. Let $\Delta^{\uparrow} = \Phi^{\uparrow}_{1}(x, k) - \Phi^{\uparrow}_0(x, k)$. Similarly, let $\Phi^{\downarrow}_0(x, k)$ and $\Phi^{\downarrow}_1(x, k)$ respectively denote the value of $\Phi^{\downarrow}(x, k)$ before and after this event. Let $\Delta^{\downarrow} = \Phi^{\downarrow}_{1}(x, k) - \Phi^{\downarrow}_0(x, k)$. We now consider four cases.

\medskip
\noindent {\em Case 1. $\Delta^{\uparrow} \geq 0$ and $\Delta^{\downarrow} \geq 0$.} In this case, we say that during this event the node $x$ {\em absorbs} $(\Delta^{\uparrow}+ \Delta^{\downarrow})$ units of potentials at level $k$.

\medskip
\noindent {\em Case 2. $\Delta^{\uparrow} < 0$ and $\Delta^{\downarrow} < 0$.} In this case, we say that during this event the node $x$ {\em releases} $-(\Delta^{\uparrow} +\Delta^{\downarrow})$ units of potential at level $k$.

\medskip
\noindent {\em Case 3. $\Delta^{\uparrow} \geq  0$ and $\Delta^{\downarrow} < 0$.} In this case, we say that during this event the node $x$ {\em absorbs} $\Delta^{\uparrow}$ units of potential at level $k$ and {\em releases} $-\Delta^{\downarrow}$ units of potential at level $k$.

\medskip
\noindent {\em Case 4. $\Delta^{\uparrow} <  0$ and $\Delta^{\downarrow} \geq 0$.} In this case, we say that during this event the node $x$ {\em releases} $-\Delta^{\uparrow}$ units of potential at level $k$ and {\em absorbs} $\Delta^{\downarrow}$ units of potential at level $k$.

\medskip
\noindent In all the above  four cases,  the {\em energy} absorbed (resp. released) by the node $x$ at level $k$ is equal to $(1+\epsilon)^{k} \cdot \epsilon^{-1}$ times the potential absorbed (resp. released) by the node $x$ at level $k$. Thus, as a matter of convention, we never allow the potential (resp. energy) released or absorbed by a node to be negative. Furthermore, during any given event, we define the  potential (resp. energy) absorbed  by a node $x$ to be the sum of the potentials (resp. energies) absorbed  by $x$ at all the levels $0 \leq k \leq L$. Similarly, the  potential (resp. energy) released by  $x$ is defined to be the sum of the potentials (resp. energies) released by $x$ at all levels $0 \leq k \leq L$. From the discussion following~(\ref{main:eq:potential:up:new})~--~(\ref{main:eq:potential:up:new:1}) and~(\ref{main:eq:potential:down:new})~--~(\ref{main:eq:potential:down:new:1}), we get the following lemma.

\begin{lemma}
\label{main:lm:work:jump}
Consider an iteration of the {\sc While} loop in Figure~\ref{main:fig:dirty:bhi15} where a node $x$ changes its level. During this iteration, the time spent in updating the data structures is at most  the energy released by the node  $x$.
\end{lemma}

Our main result is summarized in the theorem below.

\begin{theorem}
\label{main:th:amortized:work}
Starting from an empty graph, our algorithm spends $O(\tau/\epsilon^2)$ total time to handle  any sequence of $\tau$ updates (edge insertions/deletions) in $G$.   This implies an amortized update time of $O(1/\epsilon^2)$. 
\end{theorem}

We devote the rest of this section towarding giving a high level overview of the proof of the above theorem. We begin with the crucial observation that according to Lemma~\ref{main:lm:work:jump},  the energy released by the nodes  is an upper bound on the total update time of  our algorithm. Hence, in order to prove Theorem~\ref{main:th:amortized:work}, it suffices to upper bound the total energy released by the nodes during the sequence of $\tau$ updates. Note that the total energy stored at the nodes is zero when the input graph is empty.  Furthermore, the node-potentials (and energies)  are always nonnegative. Thus, during the course of our algorithm the total energy released by the nodes is at most the total energy absorbed by the nodes, and hence it suffices to upper bound the latter quantity.  We  will show that overall the nodes absorb $O(\tau/\epsilon^2)$ units of energy while our algorithm handles $\tau$ updates starting from an empty graph. This  implies Theorem~\ref{main:th:amortized:work}.  

Note that the nodes might absorb energy under two possible scenarios: 

\begin{itemize}
\item(a) An edge $(u, v)$ gets inserted into or deleted from the graph. Under this scenario, one or both the endpoints $\{u, v\}$ might absorb some energy. No node, however, changes its level under this scenario.
\item (b) The subroutine in Figure~\ref{main:fig:dirty:bhi15} is called after the insertion or deletion of an edge (scenario (a)), and a node $x$ moves up or down one level during an iteration of the {\sc While} loop in Figure~\ref{main:fig:dirty:bhi15}. Under this scenario (b), one or more nodes in $N_x \cup \{x\}$ might absorb some energy. 
\end{itemize} 
Theorem~\ref{main:th:amortized:work} now follows from Claim~\ref{main:eq:th:0} and Claim~\ref{main:eq:th:1} stated below. In the remainder of this section, we will present  high-level, intuitive justifications for each of these claims.
\begin{claim}
\label{main:eq:th:0}
The total energy absorbed by the nodes under scenario (a) is at most  $O(\tau/\epsilon)$. 
\end{claim}
\begin{claim}
\label{main:eq:th:1}
The total energy absorbed by the nodes under scenario (b) is at most  $O(\tau/\epsilon^2)$.
\end{claim}

\subsubsection{Justifying Claim~\ref{main:eq:th:0}}

Consider an {\em event} where an edge $(u, v)$ gets inserted into or deleted from the graph. This can change the potentials  of only the endpoints $u, v$, and hence only  $u$ and $v$  can absorb energy during such an event. Below, we show that the total energy absorbed by the two endpoints is at most $O(1/\epsilon)$. Since $\tau$ is the total number of edge insertions or deletions that take place in the graph $G$, this implies Claim~\ref{main:eq:th:0}.

\medskip
\noindent {\bf Edge-Deletion.}
First, we focus on analyzing an edge-deletion. Specifically, suppose that an edge $(u, v)$ with $\ell(u) = i \geq \ell(v) = j$ gets deleted from the graph. Consider the endpoint $u$. Due to this event (where the edge $(u,v)$ gets deleted), the weight $W_u$ decreases by $(1+\epsilon)^{-i}$. From~(\ref{main:eq:potential:down:new}) and~(\ref{main:eq:potential:up:new:1}) we infer that the value of $\Phi^{\downarrow}(u, i)$ can increase by at most $(1+\epsilon)^{-i}$ during this event, whereas the value of $\Phi^{\downarrow}(u, k)$ remains equal to $0$ for all $k \neq i$. In contrast, from~(\ref{main:eq:potential:up:new}) and~(\ref{main:eq:potential:down:new:1}) we infer that for all $k \in [0, L]$, the value of $\Phi^{\uparrow}(u, k)$ can never increase during this event. This is because for each level $k \in [i, L]$,  the weight $W_{u \rightarrow k}$ decreases by $(1+\epsilon)^{-i}$ due to this event. In other words, the node $u$ can only absorb at most $(1+\epsilon)^{-i}$ units of potential during this event, and that too only at level $i$. Hence, the energy absorbed by the node $u$ during this event is at most $(1+\epsilon)^{-i} \cdot (1+\epsilon)^i \cdot \epsilon^{-1} = \epsilon^{-1}$. Applying a similar argument for the other endpoint $v$, we conclude that at most $2 \epsilon^{-1} = O(1/\epsilon)$ units of energy can get absorbed due to the deletion of an edge.

\smallskip
\noindent {\bf Edge-Insertion.}
Next, we focus on the scenario where an edge $(u, v)$ gets inserted into the graph. A formal proof for this scenario is a bit involved. To highlight the main idea, we only consider one representative scenario in this section, as described below.

\medskip
\noindent {\em Suppose that $\ell(u) = i \geq \ell(v) = j$ and  $W_u > 1$ just before the insertion of the edge $(u, v)$. We want to show that the node $u$ absorbs at most $O(1/\epsilon)$ units of energy during  this event (insertion of the edge $(u,v)$).}

\medskip
\noindent The key observation here is that just before the event the node $u$ was not up-dirty. To be more specific, just before the event we had $W_{u \rightarrow i} > 1$ and $W_{u \rightarrow i+1} \leq 1$. This follows from the discussion on ``approximation ratio" in Section~\ref{sec:describe}. This implies that there was at least one edge $(u, x) \in E$ with $\ell(x) \leq i$ just before the event, for otherwise we would have $W_{u \rightarrow i+1} = W_{u \rightarrow i} > 1$. Let $i'$ be the value of $\ell^{\uparrow}(u)$ just after the event. Now, note that for every level $i \leq k \leq i'$, the value of $W_{u \rightarrow k}$ increases by $(1+\epsilon)^{-k} \leq (1+\epsilon)^{-i}$ during this event. Hence, from~(\ref{main:eq:potential:up:new}) we conclude that the node $u$ absorbs at most $(1+\epsilon)^{-i}$ units of potential at each level $k \in [i, i']$. Thus, the total energy absorbed by the node $u$ is at most $\sum_{k=i}^{i'} (1+\epsilon)^{-i} \cdot (1+\epsilon)^k \cdot \epsilon^{-1} \leq \epsilon^{-1} \cdot (1+\epsilon)^{i'-i+1}$. To complete the proof, below we show that $(1+\epsilon)^{i'-i} = O(1)$, which implies that the node $u$ absorbs at most $O(1/\epsilon)$ units of energy during this event.

Just before the event, we had $W_{u \rightarrow i+1} \leq 1$. At that time, consider a thought experiment where we move up the node $u$ to level $i'$. During that process, as we move up from level $i+1$ to level $i'$,  the weight of the edge $(u, x)$  decreases by $(1+\epsilon)^{-(i+1)} - (1+\epsilon)^{-i'}$. Hence, we get: $W_{u \rightarrow i'} \leq 1 - (1+\epsilon)^{-(i+1)} + (1+\epsilon)^{-i'}$ just before the event. Insertion of the edge $(u, v)$ increases the weight $W_{u \rightarrow i'}$ by $(1+\epsilon)^{-i'}$. Hence, we infer that $W_{u \rightarrow i'} \leq 1 - (1+\epsilon)^{-(i+1)} + 2 \cdot (1+\epsilon)^{-i'}$ just after the event. Recall that $i'$ is the value of $\ell^{\uparrow}(u)$ just after the event. Thus, by definition, we have: $W_{u \rightarrow i'} > 1$. Combining the last two inequalities, we get:
$$1 < 1 - (1+\epsilon)^{-(i+1)} + 2 \cdot (1+\epsilon)^{-i'}, \text{ which implies that } (1+\epsilon)^{-(i+1)} < 2 \cdot (1+\epsilon)^{-i'}.$$
Rearranging the terms in the last inequality, we get: $(1+\epsilon)^{i'-i} < 2 (1+\epsilon) = O(1)$, as promised.

\medskip
\noindent {\bf A note on the gap between $\ell(u)$ and $\ell^{\uparrow}(u)$.} The above argument relies upon the following property: Since the node $u$ was not up-dirty before the insertion of the edge $(u, v)$, the level $\ell^{\uparrow}(u)$ cannot be too far away from the level $\ell(u)$ just  after the insertion of the edge $(u, v)$. This property, however, might no longer be true once we call the subroutine in Figure~\ref{main:fig:dirty:bhi15}. This is because by the time we deal with a specific up-dirty node $x$, a lot of its neighbors might have changed their levels (thereby significantly changing the weight $W_x$).

\subsubsection{Justifying Claim~\ref{main:eq:th:1}}

We now give a high-level, intuitive justification for Claim~\ref{main:eq:th:1}, which bounds the total energy absorbed by all the nodes under scenario (b). See the discussion following the statement of Theorem~\ref{main:th:amortized:work}. We first classify the node-potentials into certain {\em types}, depending on the level the potential is stored at, {\em and} whether the potential will be used to account for the time spent in moving the node up or down from that level. Accordingly, for every level $k \in [0, L]$, we define: 
$$\Phi^{\uparrow}(k) = \sum_{x \in V} \Phi^{\uparrow}(x, k) \text{  and } \Phi^{\downarrow}(k) = \sum_{x \in V} \Phi^{\downarrow}(x, k).$$ 
We say that there are $\Phi^{\uparrow}(k)$ units of potential in the system that are of {\em type} $(k, \uparrow)$, and  there are $\Phi^{\downarrow}(k)$ units of potential in the system that are of {\em type} $(k, \downarrow)$. Overall, there are $2(L+1)$ different types of potentials, since we can construct $2(L+1)$ many ordered pairs of the form $(k, \alpha)$, with $0 \leq k \leq L$ and $\alpha \in \{\uparrow, \downarrow\}$.

Let $\Gamma = \{ (k, \alpha) :  \alpha \in \{\uparrow,\downarrow\} \text{ and } 0 \leq k \leq  L \}$ denote the set of all possible types of potentials. We define a total order $\succ$ on the elements of the set $\Gamma$ as follows. For any two types of potentials $(k, \alpha), (k', \alpha') \in \Gamma$, we have $(k, \alpha) \succ (k', \alpha)$ iff either $\{k > k'\}$ or $\{ k = k', \alpha = \uparrow, \text{ and } \alpha' = \downarrow \}$. Next, from~(\ref{eq:energy:up:new}) and~(\ref{eq:energy:down:new}), we recall that the {\em conversion rate} between energy and potential is $(1+\epsilon)^{k} \cdot \epsilon^{-1}$ at level $k$. In other words, from $\delta$ units of potential stored at level $k$ we get $\delta \cdot (1+\epsilon)^{k} \cdot \epsilon^{-1}$ units of energy. Keeping this in mind, we define the {\em conversion rate} associated with both the types $(k, \uparrow)$ and $(k, \downarrow)$ to be $(1+\epsilon)^k \cdot \epsilon^{-1}$. Specifically, we write $c_{(k, \uparrow)} = c_{(k, \downarrow)} = (1+\epsilon)^k \cdot \epsilon^{-1}$. Thus, from $\delta$ units of potential of any type $\gamma \in \Gamma$, we get $\delta \cdot c_{\gamma}$ units of energy. The total order $\succ$ we defined on the set $\Gamma$ has the following properties.
\begin{property}
\label{main:prop}
Consider any three types of potentials $\gamma_1, \gamma_2, \gamma_3 \in \Gamma$ such that $\gamma_1 \succ \gamma_2 \succ \gamma_3$. Then we must have $c_{\gamma_1} \geq (1+\epsilon) \cdot c_{\gamma_3}$. In words, the conversion rate between energy and potentials drops by at least a factor of $(1+\epsilon)$ as we move two hops down the total order $\Gamma$.
\end{property}

\begin{proof}
Any $\gamma \in \Gamma$ is of the form $(k, \alpha)$ where $k \in \{0, \ldots, L\}$ and $\alpha \in \{\uparrow, \downarrow\}$. From the way we have defined the total order $\succ$, it follows that if $(k_1, \alpha_1) \succ (k_2, \alpha_2) \succ (k_3, \alpha_3)$, then $k_1 \geq k_3+1$. The property holds since $c_{(k_1, \alpha_1)} = (1+\epsilon)^{k_1}$ and $c_{(k_3, \alpha_3)} = (1+\epsilon)^{k_3}$.
\end{proof}

\begin{property}
\label{main:ob:key} Consider an event where some nonzero units of potential are absorbed by some nodes, under scenario (b). Such an event occurs only if some node $x$ moves up or down one level from its current level $k$ (say). In the former case let $\gamma^* = (k, \uparrow)$, and in the latter case let $\gamma^* = (k, \downarrow)$. Let $\delta^* \geq 0$ denote the potential of type $\gamma^*$ {\em released} by $x$ during this event. For every type $\gamma \in \Gamma$, let $\delta_{\gamma} \geq 0$ denote the total potential of type $\gamma$ {\em absorbed} by all the nodes during this event. Then: (1) For every $\gamma \in \Gamma$, we have $\delta_{\gamma} > 0$ only if $\gamma^* \succ \gamma$. (2) We also have $\sum_{\gamma \in \Gamma} \delta_{\gamma} \leq \delta^*$. 
\end{property}

\begin{proof}(Sketch)
A formal proof of Property~\ref{main:ob:key} is quite involved. Instead, here we present a very high-level intuition behind the proof.  Consider an event where an up-dirty node $x$ with weight $W_x > 1$ moves up (say) from level $k$ to level $k+1$. From~(\ref{main:eq:potential:up:new}), we infer that the node $x$ releases $\delta^* = W_{x \rightarrow k} - W_{x \rightarrow k+1}$ units of potential, and the released potential is of type $(k, \uparrow)$. 

Next, we observe that the weight $W_x$ also decreases exactly by $\delta^*$ during this event. Now, during the same event, some neighbors $y$ of $x$ might decrease their weights $W_y$, and these are also the nodes that might absorb some nonzero units of potentials. We note that the weight of an edge $(x, y)$ decreases  only if $y \in N_x(0,k)$, and the sum of these weight-decreases is equal to $\delta^*$. Thus, a neighbor $y$ of $x$ absorbs some potential only if $\ell(y) \leq k$, and the sum of these absorbed potentials is at most $\delta^*$. From~(\ref{main:eq:potential:up:new}),~(\ref{main:eq:potential:up:new:1}),~(\ref{main:eq:potential:down:new}) and~(\ref{main:eq:potential:down:new:1}), we also conclude that if a node $y$ absorbs potential when its weight decreases, then the absorbed potential must be of type $(\ell(y), \downarrow)$ where $\ell(y) \leq k$. To summarize, the node $x$ releases $\delta^*$ units of potential of type $(k, \uparrow)$, and at most $\delta^*$ units of potential are overall absorbed by all the nodes during this event. Furthermore, if a nonzero amount of potential of some type $\gamma$ gets absorbed, then we must have $\gamma = (k', \downarrow)$ for some $k' \leq k$, and hence $(k, \uparrow) \succ \gamma$. 

A similar argument applies when the node $x$ moves down from level $k$ to level $k-1$.
\end{proof}

Properties~\ref{main:prop} and~\ref{main:ob:key} together give us a complete picture of the way the potential  stored by the nodes {\em flows within the system}. Specifically, there are two scenarios in which potential can be pumped into (i.e., absorbed by) the nodes. In scenario (a), potential gets pumped into the nodes {\em exogenously} by an adversary, due to the insertion or deletion of an edge in the graph. But, according to Claim~\ref{main:eq:th:0}, the total energy absorbed by the nodes under this scenario is already upper bounded by $O(\tau/\epsilon)$. On the other hand, in scenario (b), a node releases some $\delta^* \geq 0$ units of potential of type $\gamma^* \in \Gamma$ (say), and at the same time  some $0 \leq \delta \leq \delta^*$ units of potential get created. This newly created $\delta$ units of potential are then {\em split up} in some {\em chunks}, and  these chunks in turn get absorbed as potentials of (one or more) different types $\gamma$. We now note three key points about this process: (1) $\delta \leq \delta^*$. (2) If a chunk of this newly created $\delta$ units of potential gets absorbed as potential of type $\gamma$, then we must have $\gamma^* \succ \gamma$. (3) By Property~\ref{main:prop}, as we move down two hops in the total order $\succ$, the conversion rate between energy and potential {\em drops} at least by a factor of $(1+\epsilon)$. These three points together imply that the energy absorbed under scenario (b) is at most $\beta$ times the energy absorbed under scenario (a), where $\beta = 2 + 2 \cdot (1+\epsilon)^{-1} + 2 \cdot (1+\epsilon)^{-2} + \cdots = O(1/\epsilon)$. Hence, from Claim~\ref{main:eq:th:0} we infer that the total energy absorbed by the nodes under scenario (b) is at most $O(1/\epsilon) \cdot O(\tau/\epsilon) = O(\tau/\epsilon^2)$. This concludes the proof of Claim~\ref{main:eq:th:1}.

\section{Bibliography}

\bibliographystyle{plain}
\bibliography{paper}

\newpage

\appendix

\part{Full Version}
\label{part:full}

We emphasize that a few notations and definitions used in the full version (Part~\ref{part:full}) are different from the ones used in Part~\ref{part:main}. For instance, the notions of {\em up-dirty} and {\em down-dirty} nodes as defined in Section~\ref{sec:describe} are different from the ones introduced in Section~\ref{sec:algorithm}. However, a moment's thought will reveal that the algorithm in Part~\ref{part:full} is the same as the algorithm in Section~\ref{sec:describe}. This is because the {\sc While} loop in Figure~\ref{fig:dirty} moves a node $x$ to a higher (resp. lower) level iff $x$ is active and up-dirty  (resp. active and down-dirty). Furthermore, a node $x$ is active and up-dirty (resp. active and down-dirty) in Section~\ref{sec:algorithm} iff it is up-dirty (resp. down-dirty) in Section~\ref{sec:describe}. Hence, the two subroutines in Figure~\ref{main:fig:dirty:bhi15} and Figure~\ref{fig:dirty} behave in exactly the same manner. In general, the presentation in Part~\ref{part:full} is self-contained, and the notations used here should be interpreted independently of the notations used in Part~\ref{part:main}.


\section{Full version of the algorithm}
\label{sec:algorithm}

The input graph $G = (V, E)$ has $|V| = n$ nodes. Each node $v \in V$ is placed at an  integer level $\ell(v) \in \{0, \ldots, L\}$ where $L = \lceil \log_{1+\epsilon} n \rceil$. The level of an edge $(u, v) \in E$ is defined as $\ell(u, v) = \max(\ell(u), \ell(v))$. Thus, it is the maximum level among its endpoints. Each edge $(u, v) \in E$ gets a weight $w(u, v) = (1+\epsilon)^{-\ell(u,v)}$. The weight of a node $v \in V$ is given by $W_v = \sum_{u \in N_v} w(u, v)$, where $N_v$ is the set of neighbors of $v$. For any two levels $0 \leq i \leq j \leq L$, let $N_v(i,j) = \{ u \in N_v : i \leq \ell(u) \leq j \}$ denote  the set of neighbors of a node $v \in V$ who lie between level $i$ and level $j$. 

For every node $v \in V$ and every level $0 \leq i \leq L$, let $W_{v \rightarrow i} = \sum_{u \in N_v} (1+\epsilon)^{-\max(\ell(u), i)}$ denote what  the weight of the node $v$ would have been if we were to place  $v$ at level $i$, without changing the level of any other node. We have $W_v = W_{v \rightarrow \ell(v)}$. Note that $W_{v \rightarrow i}$ is a monotonically (weakly) decreasing function of $i$. This holds since as we increase the level of $v$ (say) from $i$ to $(i+1)$, all its incident edges $(u, v) \in E$ with $u \in N_v(0, i)$ decrease their weights, and the weights of all its other incident edges remain unchanged. We will repeatedly use this observation throughout the rest of this paper.

\paragraph{A classification of nodes.}
 We  classify the nodes into two types:  {\em up-dirty} and  {\em down-dirty}. A node $v \in V$ is  up-dirty if $W_v \geq 1$, and   down-dirty if $W_v < 1$. Note that an up-dirty node can never be at  level $L$.
 \begin{corollary}
\label{cor:updirty:level}
Every up-dirty node $v \in V$ has $\ell(v) < L$.
\end{corollary}

\begin{proof}
A node $v \in V$ can have at most $(n-1)$ neighbors in an $n$-node graph, and  the weight of each of its incident edges $(u, v) \in E$ becomes equal to $(1+\epsilon)^{-L} \leq 1/n$ if we place the node $v$ at level $L$. Hence, we have $W_{v \rightarrow L} \leq |N_v| \cdot (1/n) \leq (n-1)/n < 1$. Thus, a node $v$ can never be up-dirty if it is at level $L$.
\end{proof}

We further classify the nodes into two types:  {\em active} and {\em passive}. An up-dirty node $v \in V$ is  active if $W_{v \rightarrow \ell(v)+1} \geq 1$, and  passive otherwise.  In contrast, a down-dirty node $v \in V$ is active if $\{ \ell(v) > 0 \text{ and } W_v < 1-\epsilon\}$, and passive otherwise. We now derive an important corollary.
\begin{corollary}
\label{cor:active}
If a up-dirty node $v \in V$ is active, then we have $W_{v \rightarrow \ell(v)+1} \geq 1$. If a down-dirty node $v \in V$ is active, then we have $W_{v \rightarrow \ell(v) -1} < 1$. In words, an up-dirty (resp. down-dirty) node is active only if it remains up-dirty (resp. down-dirty) after we increase (resp. decrease) its level by one.
\end{corollary}

\begin{proof}
The first part of the corollary follows from the definition of an active, up-dirty node. Accordingly, we consider an active, down-dirty node $v \in V$ at level $\ell(v) = k$. By definition, we have $k > 0$ and $W_{v \rightarrow k} < 1-\epsilon$. Now, if we move the node $v$ from level $k$ to level $k-1$, then the weight of every edge in $(u,v) \in E$ with $u \in N_v(0,k-1)$ increases by a factor of $(1+\epsilon)$, and the weight of every other edge remains unchanged. This ensures that $W_{v \rightarrow k-1} \leq (1+\epsilon) \cdot W_{v \rightarrow k} < (1+\epsilon) \cdot (1-\epsilon)  < 1$.
\end{proof}

In the dynamic setting, when the input graph $G = (V, E)$ keeps getting updated via a sequence of edge insertions and deletions, our algorithm will strive to maintain the following invariant.

\begin{invariant}
\label{inv}
Every   node $v \in V$ is passive.
\end{invariant}

In Corollary~\ref{cor:inv}, we will show how to get a $(2+\epsilon)$-approximate minimum vertex cover in $G$ out of Invariant~\ref{inv}. But first, we observe an useful bound on the weights of the passive nodes.

\begin{corollary}
\label{cor:passive}
For every passive node $v \in V$ at level $\ell(v) > 0$, we have $(1-\epsilon) \leq W_v < (1+\epsilon)$. Furthermore, for every passive node $v \in V$ at level $\ell(v) = 0$, we have $0 \leq W_v < (1+\epsilon)$.
\end{corollary}

\begin{proof}
Consider a node $v \in V$ at level $\ell(v) = k > 0$. First, for the sake of contradiction, suppose that  $W_v \geq 1+\epsilon$. Then clearly the node $v$ is up-dirty. If we increase the level of $v$ from $k$ to $k+1$, then the weight of every edge $(v, x) \in E$ with $x \in N_v(0, k)$ decreases by a factor of $(1+\epsilon)$, and the weight of every other edge remains unchanged. It follows that $W_{v \rightarrow k+1} \geq (1+\epsilon)^{-1} \cdot W_{v \rightarrow k} \geq 1$. Hence, such a node $v$ must also be active, which contradicts the assumption stated in the corollary. Next, for the sake of contradiction, suppose that $W_v < 1-\epsilon$. Then by definition the node $v$ is down-dirty and active, which again contradictions the assumption stated in the corollary. We therefore conclude that if a node $v$ is passive at a level $\ell(v) > 0$, then we must have $(1-\epsilon) \leq W_v < (1+\epsilon)$.

Next, consider a node $v \in V$ at level $\ell(v) = 0$. If $W_v \geq 1+\epsilon$, then the node $v$ is up-dirty, and using a similar argument described in the previous paragraph, we can derive that the node $v$ is also active. This leads to a contradiction. Thus, a passive node $v$ at level $\ell(v) = 0$ must have $0 \leq W_v < (1+\epsilon)$.
\end{proof}

\begin{corollary}
\label{cor:inv}
Under Invariant~\ref{inv}, the set of nodes $V^* = \{ v \in V : W_v \geq (1-\epsilon)\}$ forms a $2(1+\epsilon)(1-\epsilon)^{-1}$-approximate minimum vertex cover in $G = (V, E)$.
\end{corollary}

\begin{proof}(Sketch) From Corollary~\ref{cor:passive}, it follows that under Invariant~\ref{inv} every node $v \in V$ at level $\ell(v) > 0$ belongs to the set $V^*$. Now, consider any two nodes $x, y \in V \setminus V^*$. Clearly, we have $\ell(x) = \ell(y) = 0$ and $W_x, W_y < (1-\epsilon)$. There cannot be any edge between these two nodes $x$ and $y$, for otherwise such an edge $(x, y) \in E$ will have weight $w(x, y) = 1$, which in turn will contradict the assumption that $W_x, W_y < (1-\epsilon)$. Thus, we conclude that there cannot be any edge between two nodes in $V \setminus V^*$. In other words, the set $V^*$ forms a valid vertex cover.

Now, we scale the edge-weights by a factor of $(1+\epsilon)$ by setting $\tilde{w}(e) = (1+\epsilon)^{-1} \cdot w(e)$ for all $e \in E$.  Let $\tilde{W}_v = \sum_{u \in N_v} \tilde{w}(u,v)$ denote the weight received by a node $v \in V$ from its incident edges under these scaled edge-weights. Corollary~\ref{cor:passive} and Invariant~\ref{inv} imply that $0 \leq \tilde{W}_v = (1+\epsilon)^{-1} \cdot W_v < 1$ for all nodes $v \in V$. In other words, the  edge-weights $\{\tilde{w}(e)\}$ form a valid fractional matching in the graph $G = (V, E)$. Note that a node $v \in V$ belongs to the set $V^*$ if and only if $(1-\epsilon) (1+\epsilon)^{-1} \leq \tilde{W}_v = (1+\epsilon)^{-1} \cdot W_v < 1$. 

To summarize, we have constructed a valid fractional matching $\{\tilde{w}(e)\}$ and a valid vertex cover $V^*$ with the following property: Every node in $V^*$ receives a weight between $(1-\epsilon)(1+\epsilon)^{-1}$ and $1$ under the  matching $\{ \tilde{w}(e) \}$. Hence, from the complementary slackness conditions between the pair of primal and dual LPs for minimum fractional vertex cover and maximum fractional matching, we can infer  that the set $V^*$ forms a $2(1+\epsilon)(1-\epsilon)^{-1}$-approximate minimum vertex cover in $G$.
\end{proof}

\paragraph{Handling the insertion/deletion of an edge:} Initially, the  graph $G = (V, E)$ has an empty edge-set, every node $v \in V$ is at level $\ell(v) = 0$, and Invariant~\ref{inv} is trivially satisfied. By inductive hypothesis, suppose that  Invariant~\ref{inv} is satisfied just before the insertion or deletion of an edge in $G$. Now, as an edge $(u, v)$ gets inserted into (resp. deleted from) the graph, the node-weights $W_u$ and $W_v$ will both increase (resp. decrease) by $(1+\epsilon)^{-\max(\ell(u), \ell(v))}$, and one or both of the endpoints $x \in \{u, v\}$ might become active, thereby violating Invariant~\ref{inv}. If this is the case, then we  call the subroutine outlined in Figure~\ref{fig:dirty}. The subroutine returns only after ensuring that Invariant~\ref{inv} is satisfied. At that point, we are  ready to handle the next edge insertion or deletion in $G$.
\begin{figure}[h!]
	\centerline{\framebox{
			\begin{minipage}{5.5in}
				\begin{tabbing}
					1. \=    {\sc While}  there exists some active  node $x$: \\
					2.  \= \ \ \ \ \ \ \= {\sc If} the node $x$ is active and up-dirty, {\sc then} \\
					3. \> \> \ \ \ \ \ \ \ \= move it up by one level by setting $\ell(x) \leftarrow \ell(x) + 1$. \\
					4. \> \> {\sc Else if} the node $x$ is active and down-dirty, {\sc then} \\
					5. \> \> \> move it down one level by setting $\ell(x) \leftarrow \ell(x) - 1$.
					\end{tabbing}
			\end{minipage}
	}}
	\caption{\label{fig:dirty} Subroutine:  FIX-DIRTY(.) is called after the insertion/deletion of an edge.}
\end{figure}

\subsection{Bounding the work done by our algorithm} 
\label{sec:analyze:work}
 
We say that  our algorithm performs one unit of {\em work} each time it  changes the level $\ell(y, z) = \max(\ell(y), \ell(z))$ of an edge $(y, z) \in E$. To be more specific, consider an iteration of the {\sc While} loop in Figure~\ref{fig:dirty} where we are moving up a node $x$ from (say) level $k$ to level $k-1$. During this iteration, the edges that change their levels are of the form $(x, y) \in E$ with $y \in N_x(0, k)$. Hence, the work done by our algorithm during such an iteration is equal to $|N_x(0, k)|$, and we {\em charge} this work on the node $x$. In words, we say that our algorithm performs $|N_x(0,k)|$ units of work {\em on} the node $x$ during this iteration. Similarly, consider an iteration of the {\sc While} loop in Figure~\ref{fig:dirty} where we are moving down a node $x$ from (say) level $k$ to level $k-1$. During this iteration, the edges that change their levels are of the form $(x, y) \in E$ with $y \in N_x(0,k-1)$. Hence, the work done by our algorithm during such an iteration is equal to $|N_x(0,k-1)|$, and we {\em charge} this work on the node $x$. In words, we say that our algorithm performs $|N_x(0,k-1)|$ {\em on} the node $x$ during this iteration.

The work done  serves as a useful {\em proxy} for the  update time. We will show that starting from an empty graph, we  perform $O(\tau/\epsilon^2)$ units of work to handle a sequence of $\tau$ edge insertions/deletions in $G$. Later on, in Section~\ref{sec:analyze:time}, we will explain why this implies an amortized update time of $O(1/\epsilon^2)$, after describing a set of suitable data structures for implementing our algorithm.

\paragraph{Node-potentials and energy.} In order to bound the work done by our algorithm, we introduce the notions of the {\em potential} and {\em energy} stored by a  node. Every node $x \in V$ stores  nonnegative amounts of potential $\Phi(x, k) \geq 0$ and energy $\E(x, k) = \Phi(x, k) \cdot (1+\epsilon)^k \cdot \epsilon^{-1}$ at every level $0 \leq k \leq L$. Intuitively, when the node $x$ moves up (or down) by one level from the level $k$, it {\em releases} the potential (and  energy) stored by it at that level. We will show that the amount of energy released is sufficient to account for the work done on the node $x$ during this event.  Our definitions of node-potentials and energy are, therefore, intimately tied with the work done by our algorithm to move up (or down) a node by one level. The next two claims give explicit expressions for this quantity. 

\begin{claim}
\label{cl:motivate:potential:up}
Consider an event where we increase the level of a node $x$ (say) from $k$ to $k+1$. The work done on the node $x$ during this event is equal to $\left(W_{x \rightarrow k} - W_{x \rightarrow k+1}\right) \cdot (1+\epsilon)^{k+1} \cdot \epsilon^{-1}$.
\end{claim}

\begin{proof}
During this event, the work done on the node $x$ is equal to the number of edges that change their levels, which in turn is equal to $N_x(0,k)$. Each edge $(x, y) \in E$ with $y \in N_x(0, k)$ changes its weight from $(1+\epsilon)^{-k}$ to $(1+\epsilon)^{-k-1}$ as the node $x$ moves up from level $k$ to level $k+1$. Hence, we have $W_{x \rightarrow k} - W_{x \rightarrow k+1} = |N_x(0,k)| \cdot \left( (1+\epsilon)^{-k} - (1+\epsilon)^{-k-1}\right) = |N_x(0,k)| \cdot \epsilon \cdot (1+\epsilon)^{-k-1}$. Rearranging the terms in this equality, we infer that the work done on the node $x$ during this event is equal to: $|N_x(0, k)| = \left( W_{v \rightarrow k} - W_{v \rightarrow k+1} \right) \cdot (1+\epsilon)^{k+1} \cdot \epsilon^{-1}$.
\end{proof}

\begin{claim}
\label{cl:motivate:potential:down}
Consider an event where we decrease the level of a node $x$ (say) from $k$ to $k-1$. The work done on the node $x$ during this event is equal to $\left(W_{x \rightarrow k-1} - W_{x \rightarrow k}\right) \cdot (1+\epsilon)^k \cdot \epsilon^{-1}$.
\end{claim}

\begin{proof}
During this event, the work done on the node $x$ is equal to the number of edges that change their levels, which in turn is equal to $N_x(0,k-1)$. Each edge $(x, y) \in E$ with $y \in N_x(0, k-1)$ changes its weight from $(1+\epsilon)^{-k}$ to $(1+\epsilon)^{-k+1}$ as the node $x$ moves down from level $k$ to level $k-1$. Hence, we have $W_{x \rightarrow k-1} - W_{x \rightarrow k} = |N_x(0,k-1)| \cdot \left( (1+\epsilon)^{-k} - (1+\epsilon)^{-k+1}\right) = |N_x(0,k-1)| \cdot \epsilon \cdot (1+\epsilon)^{-k}$. Rearranging the terms in this equality, we infer that the work done on the node $x$ during this event is equal to: $|N_x(0, k-1)| = \left( W_{v \rightarrow k} - W_{v \rightarrow k+1} \right) \cdot (1+\epsilon)^{k} \cdot \epsilon^{-1}$.
\end{proof}

Before proceeding any further, we  formally clarify the way we are going to use  two phrases: {\em potential (resp. energy) absorbed by a node}, and {\em potential (resp. energy) released by a node}. Fix a node $x \in V$ and a level $0 \leq k \leq L$, and consider an event which changes the potential $\Phi(x, k)$. Let $\Phi^{(0)}(x, k)$ and $\Phi^{(1)}(x, k)$ respectively denote the value of $\Phi(x, k)$ before and after the event, and let $\Delta = \Phi^{(1)}(x, k) - \Phi^{(0)}(x, k)$. We now consider four possible situations.
\begin{enumerate}
\item {\em The node $x$ does not change its state from up-dirty to down-dirty (or vice versa) due to the event.} Thus, either the node $x$ remains up-dirty both before and after the event, or the node remains down-dirty both before and after the event. Here, if $\Delta > 0$, then we say that the node $x$ {\em absorbs} $\Delta$ units of potential and {\em releases} zero unit of potential at level $k$ due to this event. Otherwise, if $\Delta < 0$, then we say that the node $x$ {\em absorbs} zero unit of potential and {\em releases} $-\Delta$ units of potential at level $k$ due to this event. Finally, if $\Delta = 0$, then we say that the node {\em absorbs} and {\em releases} zero unit of potential at level $k$ due to this event. 
\item {\em The node $x$ changes its state due to the event.} Thus, either the node $x$ is up-dirty before and down-dirty after the event, or the node $x$ is down-dirty before and up-dirty after the event. Here, we say that the node $x$ {\em releases} $\Phi^{(0)}(x, k)$ units of potential at level $k$ due to the event, and it {\em absorbs} $\Phi^{(1)}(x, k)$ units of potential at level $k$ due to the event.
\end{enumerate}
In both situations, we say that due to the event the energy absorbed (resp. released) by the node $x$ at level $k$ is $(1+\epsilon)^{k} \cdot \epsilon^{-1}$ times the potential absorbed (resp. released) by the node $x$ at level $k$. Note that as a matter of convention, we never allow the potential (resp. energy) released or absorbed by a node to be negative.  During any given event, the total potential (resp. energy) absorbed  by a node $x$ is the sum of the potentials (resp. energies) absorbed  by $x$ at all the levels $0 \leq k \leq L$, and the total potential (resp. energy) released by  $x$ is the sum of the potentials (resp. energies) released by $x$ at all levels $0 \leq k \leq L$. Thus, during any given event, according to our convention the {\em net} increase in the value of $\Phi(x)$ (resp. $\E(x)$)  is equal to the total potential (resp. energy) absorbed by  $x$ {\em minus} the total potential (resp. energy) released by  $x$. 

We define the node-potentials and energy in such a way which ensures the following property: At any point in time during the course of our algorithm, the work done on a node $x$ is at most $(1+\epsilon)$ times the energy released by $x$. This is stated in the lemma below. 

\begin{lemma}
\label{lm:work:jump}
Consider any iteration of the {\sc While} loop in Figure~\ref{fig:dirty} where a node $x \in V$ changes its level by one. During the concerned iteration of the {\sc While} loop, the work done  on the node $x$  is at most $(1+\epsilon)$ times the energy released by  $x$.
\end{lemma}

We first define the potentials and energy stored by an up-dirty node.  For any up-dirty node $x \in V$, let $\ell(x, \uparrow)$ to be the maximum level $k \in \{\ell(x), \ldots, L\}$ where $W_{x \rightarrow k} \geq 1$. The potential $\Phi(x, k)$ and the energy $\E(x, k)$ are then defined as follows.
\begin{eqnarray}
\label{eq:potential:up:new}
\text{If a node } x  \text{ is up-dirty, then } \begin{cases}   \Phi(x, k) = \begin{cases} 0 & \text{ for all } \ell(x, \uparrow) < k \leq L; \\
W_{x \rightarrow k} - 1 & \text{ for } k = \ell(x, \uparrow); \\
W_{x \rightarrow k} - W_{v \rightarrow k+1} & \text{ for all }  \ell(v) \leq  k < \ell(x, \uparrow); \\
0 & \text{ for all } 0 \leq k < \ell(x).
\end{cases} \\ \\
\label{eq:energy:up:new}
 \E(x, k) = \Phi(x, k) \cdot (1+\epsilon)^k \cdot \epsilon^{-1} \text{ at  every level } 0 \leq k \leq L. 
 \end{cases}
\end{eqnarray}
To gain some intuition about the definition of this node-potential, recall that $W_{x \rightarrow k}$ is a monotonically (weakly) decreasing function of $k$. Since the node $x$ is up-dirty, we have: 
\begin{equation}
\label{eq:explain:up}
W_{x \rightarrow \ell(x)} \geq \cdots \geq W_{x \rightarrow \ell(x, \uparrow)}  \geq 1 >  W_{x \rightarrow \ell(x, \uparrow)+1}
\end{equation}
An immediate corollary of~(\ref{eq:explain:up}) is that  $\Phi(x, k) \geq 0$ for any up-dirty node $x$ and any level $k$. Now, suppose that currently $\ell(x) = i$ and the node $x$ is {\em active} and up-dirty, so that $i \leq \ell(x, \uparrow) - 1$,  and consider the event where the {\sc While} loop in Figure~\ref{fig:dirty} is moving the node $x$ up from level $i$ to level $i+1$. Note that $\Phi(x, i) = W_{x \rightarrow i} - W_{x \rightarrow i+1}$ and $\E(x, i) = \left(W_{x \rightarrow i} - W_{x \rightarrow i+1}\right) \cdot (1+\epsilon)^{i} \cdot \epsilon^{-1}$ before the event, and $\Phi(x, i) = \E(x, i) = 0$ after the event. Hence, during this event the node $x$ {\em releases} $\left(W_{x \rightarrow i} - W_{x \rightarrow i+1}\right)$ units of potential and $\left(W_{x \rightarrow i} - W_{x \rightarrow i+1}\right) \cdot (1+\epsilon)^{i} \cdot \epsilon^{-1}$ units of energy at level $i$. From Claim~\ref{cl:motivate:potential:up}, we conclude that during this event the work done on the node $x$ is at most $(1+\epsilon)$ times the energy released by the node $x$. This proves Lemma~\ref{lm:work:jump} when the node $x$ is up-dirty.
\begin{corollary}
\label{cor:total:potential:up:new}
Every up-dirty node $x$ has $\Phi(x) = \sum_{k=0}^L \Phi(x, k) = W_{x \rightarrow \ell(x)} - 1 = W_x - 1$.
\end{corollary}
\begin{proof}
Follows if we sum the values of $\Phi(x, k)$, as defined in~(\ref{eq:potential:up:new}), over all levels $0 \leq k \leq L$.
\end{proof}

We now define the potentials and energy for a down-dirty node.
\begin{eqnarray}
\label{eq:potential:down:new}
\text{If a node } x  \text{ is down-dirty, then }  \begin{cases} \Phi(x, k) = \begin{cases}  0 & \text{ for all } \ell(x) < k \leq L; \\
 1 - W_{x \rightarrow k}  & \text{ for }  k = \ell(x); \\
0 & \text{ for all } 0 \leq k  < \ell(x).
\end{cases} \\ \\
\label{eq:energy:down:new}
 \E(x, k) = \Phi(x, k) \cdot (1+\epsilon)^k \cdot \epsilon^{-1}  \text{ at every level } 0 \leq k \leq L.
 \end{cases}
\end{eqnarray}
An immediate corollary of~(\ref{eq:potential:down:new}) is that  $\Phi(x, k) \geq 0$ for any down-dirty node $x$ and any level $k$. Now, suppose that currently the node $x$ is {\em active} and down-dirty at level $\ell(x) = i$, so that by Corollary~\ref{cor:active} we have $1 > W_{x \rightarrow k-1} \geq W_{x \rightarrow k}$. Consider an event where the {\sc While} loop in Figure~\ref{fig:dirty} is moving the node $x$ down from level $i$ to level $i-1$. Note that $\Phi(x, i) = 1 - W_{x \rightarrow i} \geq \left(W_{x \rightarrow i-1} - W_{x \rightarrow i}\right)$ and $\E(x, i) = \Phi(x, i) \cdot (1+\epsilon)^{i} \cdot \epsilon^{-1}$ before the event, and $\Phi(x, i) = \E(x, i) = 0$ after the event. Hence, during this event the node $x$ {\em releases} at least $\left(W_{x \rightarrow i-1} - W_{x \rightarrow i}\right)$ units of potential and at least $\left(W_{x \rightarrow i-1} - W_{x \rightarrow i}\right) \cdot (1+\epsilon)^{i} \cdot \epsilon^{-1}$ units of energy at level $i$. From Claim~\ref{cl:motivate:potential:down}, we therefore conclude that during this event the work done on the node $x$ is at most  the energy released by the node $x$. This proves Lemma~\ref{lm:work:jump} when the node $x$ is down-dirty. The corollary below bounds the total potential stored by a down-dirty node.
\begin{corollary}
\label{cor:total:potential:down:new}
Every down-dirty node $x \in V$ has $\Phi(x) = \sum_{k=0}^L \Phi(x, k) = 1- W_{x \rightarrow \ell(x)}  = 1- W_x$.
\end{corollary}
\begin{proof}
Follows if we sum the values of $\Phi(x, k)$, as defined in~(\ref{eq:potential:down:new}), over all levels $0 \leq k \leq L$.
\end{proof}

Lemma~\ref{lm:work:jump} implies that in order to bound the total work done by our algorithm, all we need to do is to bound the total energy released by  the nodes. Initially, when the graph $G = (V, E)$ has an empty edge-set, every node is at level zero and has zero energy. This implies that during the course of our algorithm the total energy released by the nodes is at most the total energy {\em absorbed} by the nodes, and we will try to upper bound the latter quantity.   Note that a node $x \in V$ absorbs or releases nonzero potential (and energy) only if one of the following three events takes place. (1) An edge $(x, y)$ incident on $x$ gets inserted into or deleted from the graph $G = (V, E)$. (2) During an iteration of the {\sc While} loop in Figure~\ref{fig:dirty}, the node $x$ is moved to a different level. (3) During an iteration of the {\sc While} loop in Figure~\ref{fig:dirty}, a neighbor $y$ of $x$ is moved to a different level. Such an event might change the weight of the edge $(x, y)$, which in turn might lead to a change in the energy stored by $x$. We will upper bound the energy absorbed by the nodes under all the three events (1), (2) and (3). These bounds will follow from Lemmas~\ref{lm:insertion:deletion},~\ref{lm:transfer:up:jump} and~\ref{lm:transfer:down:jump}.

\begin{lemma}
\label{lm:insertion:deletion}
Due to the insertion or deletion of an edge $(u, v)$ in the graph, each of the endpoints $x \in \{u, v\}$ absorbs at most $3 \epsilon^{-1}$ units of energy.
\end{lemma}

The  proof of Lemma~\ref{lm:insertion:deletion} appears in Appendix~\ref{sec:lm:insertion:deletion}. To gain some intuition behind the proof, consider an alternate scenario where we defined the potential and energy stored by a node as follows.
\begin{eqnarray}
\label{eq:alternate:potential:energy}
\text{For every node } x \in V, \text{ we have } \begin{cases}
\Phi^*(x, k) = \begin{cases}  |W_x - 1| &  \text{ for } k = \ell(x); \\
 0 & \text{ otherwise.} 
 \end{cases} \\ \\
\E^*(x,k)  = \Phi^*(x,k ) \cdot (1+\epsilon)^{\ell(x)} \cdot \epsilon^{-1} \text{ for all } 0 \leq k \leq L.
\end{cases}
\end{eqnarray}
To appreciate the link between these alternate notions  with the actual ones that we use, first consider an up-dirty node $x \in V$. From~(\ref{eq:potential:up:new}) and Corollary~\ref{cor:total:potential:up:new}, it follows that in our actual definition the total potential stored by the node $x$ is equal to $W_x - 1$. We {\em split up} this potential in a carefully chosen way and store a specific part resulting from this split at each level $\ell(x) \leq k \leq \ell(x, \uparrow)$. Further, the {\em conversion rate} between energy and potential is fixed at $(1+\epsilon)^k \cdot \epsilon^{-1}$ at each level $k$. In words, every $\delta$ units of potential at a level $k$ results in $\delta \cdot (1+\epsilon)^{k} \cdot \epsilon^{-1}$ units of energy. In~(\ref{eq:alternate:potential:energy}), all we are doing is to get rid of this {\em split}. Now, all the potential $W_x -1$ is stored at level $\ell(x)$, and the conversion rate between energy and potential remains the same as before. For a down-dirty node, in contrast, note that the definition~(\ref{eq:alternate:potential:energy}) remains the same as in~(\ref{eq:potential:down:new}).

It is easy to see why Lemma~\ref{lm:insertion:deletion} holds in the alternate scenario. Consider the insertion or deletion of an edge $(u, v)$ and an endpoint $x \in \{u, v\}$ of this edge. We have $w(u, v) = (1+\epsilon)^{-\max(\ell(u), \ell(v))} \leq (1+\epsilon)^{-\ell(x)}$.  Hence, due to the insertion or deletion of the edge $(u, v)$, the potential $\Phi^*(x, \ell(x))$ can increase by at most $(1+\epsilon)^{-\ell(x)}$ and the node $x$ can absorb at most $(1+\epsilon)^{-\ell(x)} \cdot (1+\epsilon)^{\ell(x)} \cdot \epsilon^{-1} = \epsilon^{-1}$ units of energy. Now, if we want to show that Lemma~\ref{lm:insertion:deletion} holds even with the {\em actual} notions of potential and energy actually used in our analysis, then the proof becomes a bit more complicated. But the main idea behind the proof is to apply the same argument, along with the additional observation that the conversion rate between energy and potential increases geometrically in powers of $(1+\epsilon)$ as we move up to higher and higher levels.  The next two lemmas bound the potentials absorbed by all the nodes as some node $x$ moves up or down by one level. Their proofs appear in Appendix~\ref{sec:lm:transfer:up:jump} and Appendix~\ref{sec:lm:transfer:down:jump}.

\begin{lemma}
\label{lm:transfer:up:jump}
Consider an {\em up-jump}, where an active up-dirty node $x$ moves up from level $k$ (say) to level $k+1$ (as per steps 2, 3 in Figure~\ref{fig:dirty}). For every node $v \in V$ and every level $0 \leq t \leq L$, let $\delta(v, t, \rightarrow)$ and $\delta(v, t, \leftarrow)$ respectively denote the potential released and the potential absorbed by the node $v$  at level $t$ during this up-jump of $x$. Then we have:
\begin{enumerate}
\item For every node $v \in V$, we have $\delta(v, t, \leftarrow) > 0$ only if $v$ is down-dirty  after the up-jump of $x$, $v \in N_x$ and $k \geq t$.
\item $\sum_{v} \sum_t \delta(v, t, \leftarrow) \leq \delta(x, k, \rightarrow)$. In words, during the up-jump of $x$, the total potential absorbed by all the nodes  is at most the  potential released by $x$ at level $k$.
\end{enumerate}
\end{lemma}

\begin{lemma}
\label{lm:transfer:down:jump}
Consider a {\em down-jump}, where an active down-dirty node $x$ moves down from level $k$ (say) to level $k-1$ (as per steps 4, 5 in Figure~\ref{fig:dirty}). For every node $v \in V$ and every level $0 \leq t \leq L$, let $\delta(v, t, \rightarrow)$ and $\delta(v, t, \leftarrow)$ respectively denote the potential released and the potential absorbed by the node $v$  at level $t$ during this down-jump of $x$. Then we have:
\begin{enumerate}
\item  For every node $v \in V$, we have $\delta(v, t, \leftarrow) > 0$ only if $v \in N_x \cup \{x\}$ and $k > t$.
\item $\sum_{v} \sum_t \delta(v, t, \leftarrow) \leq \delta(x, k, \rightarrow)$.  In words, during the down-jump of $x$, the total potential absorbed by all the nodes  is at most the  potential released by $x$ at level $k$.
\end{enumerate}
\end{lemma}

\medskip
We are now ready to state the main theorem of our paper. Its proof appears in Section~\ref{sec:th:amortized:work}.

\begin{theorem}
\label{th:amortized:work}
Starting from an empty graph, consider a sequence of $\tau$  updates (edge insertions/deletions) in $G$. The nodes release $O(\tau/\epsilon^2)$ energy while our algorithm handles these updates.
\end{theorem}

\begin{corollary}
\label{cor:main:update:time}
Starting from an empty graph, consider a sequence of $\tau$ updates (edge insertions/deletions) in $G$. Our algorithm does $O(\tau/\epsilon^2)$ units of work while handling these updates.
\end{corollary}

\begin{proof}
Follows from Lemma~\ref{lm:work:jump} and Theorem~\ref{th:amortized:work}.
\end{proof}

\subsection{Proof of Theorem~\ref{th:amortized:work}}
\label{sec:th:amortized:work}
The total energy stored at the nodes is zero when the input graph is empty.  Furthermore, the node-potentials $\{\Phi(v, k)\}$  are always nonnegative. Hence, during the course of our algorithm, the total energy released by the nodes is at most the total energy absorbed by the nodes.  We will show that overall the nodes absorb $O(\tau/\epsilon^2)$ units of energy while our algorithm handles $\tau$ updates starting from an empty graph. This will imply the theorem. We start by observing that the nodes might absorb energy under two possible scenarios: 
\begin{itemize}
\item (a) An edge $(u, v)$ gets inserted into or deleted from the graph. Under this scenario, one or both the endpoints $\{u, v\}$ might absorb some energy. 
\item (b) A node $x$ moves up or down one level during the execution of the {\sc While} loop in Figure~\ref{fig:dirty} after the insertion/deletion of an edge. Under this scenario, one or more nodes in $N_x \cup \{x\}$ might absorb some energy.  Note that in the former case (where $x$ is moving up from level $k$) the node $x$ is {\em active} and up-dirty at level $k$, and in the latter case (where $x$ is moving down from level $k$) the node $x$ is {\em active} and down-dirty at level $k$. This follows from Figure~\ref{fig:dirty}.
\end{itemize}
We will show that:
\begin{equation}
\label{eq:th:0}
\text{The total energy absorbed by the nodes under scenario (a) is at most } O(\tau/\epsilon).
\end{equation}
\begin{equation}
\label{eq:th:1}
\text{The total energy absorbed by the nodes under scenario (b) is at most } O(\tau/\epsilon^2).
\end{equation}
The theorem follows from~(\ref{eq:th:0}) and~(\ref{eq:th:1}). We immediately note that~(\ref{eq:th:0}) holds due to Lemma~\ref{lm:insertion:deletion}. Accordingly, from now onward, we focus on proving~(\ref{eq:th:1}). We start by introducing a few important notations and terminologies, in order to set the stage up for the actual analysis.

Let $\Phi(k) = \sum_v \Phi(v, k)$ denote the  potential stored at level $k$ by all the nodes. Further, let $\Phi(k, \uparrow) = \sum_{v : v \text{ is up-dirty}} \Phi(v, k)$ and $\Phi(k, \downarrow) = \sum_{v  : v \text{ is down-dirty}} \Phi(v, k)$ respectively denote the potential stored at level $k$ by all the up-dirty and down-dirty nodes. As every node is either up-dirty or down-dirty (but not both at the same time), we have $\Phi(k) = \Phi(k, \uparrow) + \Phi(k, \downarrow)$. Informally, we say that there are $\Phi(k, \uparrow)$ units of potential in the system that are of {\em type} $(k, \uparrow)$, and  there are $\Phi(k, \downarrow)$ units of potential in the system that are of {\em type} $(k, \downarrow)$. There are $2(L+1)$ different types of potentials, for there are $2(L+1)$ many ordered pairs of the form $(k, \alpha)$, with $0 \leq k \leq L$ and $\alpha \in \{\uparrow, \downarrow\}$.

Let $\Gamma = \{ (k, \alpha) :  \alpha \in \{\uparrow,\downarrow\} \text{ and } 0 \leq k \leq  L \}$ denote the set of all possible types of potentials. We define a total order $\succ$ on the elements of the set $\Gamma$ as follows. For any two types of potentials $(k, \alpha), (k', \alpha') \in \Gamma$, we have $(k, \alpha) \succ (k', \alpha)$ iff either $\{k > k'\}$ or $\{ k = k', \alpha = \uparrow, \alpha' = \downarrow \}$. Next, from~(\ref{eq:energy:up:new}) and~(\ref{eq:energy:down:new}), we recall that the {\em conversion rate} between energy and potential is $(1+\epsilon)^{k} \cdot \epsilon^{-1}$ at level $k$. In other words, from $\delta$ units of potential stored at level $k$ we get $\delta \cdot (1+\epsilon)^{k} \cdot \epsilon^{-1}$ units of energy. Keeping this in mind, we define the {\em conversion rate} associated with both the types $(k, \uparrow)$ and $(k, \downarrow)$ to be $(1+\epsilon)^k \cdot \epsilon^{-1}$. Specifically, we write $c_{(k, \uparrow)} = c_{(k, \downarrow)} = (1+\epsilon)^k \cdot \epsilon^{-1}$. Thus, from $\delta$ units of potential of any type $\gamma \in \Gamma$, we get $\delta \cdot c_{\gamma}$ units of energy. The total order $\succ$ we defined on the set $\Gamma$ has the following property that will be crucially used in our analysis.
\begin{property}
\label{prop}
Consider any three types of potentials $\gamma_1, \gamma_2, \gamma_3 \in \Gamma$ such that $\gamma_1 \succ \gamma_2 \succ \gamma_3$. Then we must have $c_{\gamma_1} \geq (1+\epsilon) \cdot c_{\gamma_3}$. In words, the conversion rate between energy and potentials drops by at least a factor of $(1+\epsilon)$ as we move two hops down the total order $\Gamma$.
\end{property}

\begin{proof}
Any $\gamma \in \Gamma$ is of the form $(k, \alpha)$ where $k \in \{0, \ldots, L\}$ and $\alpha \in \{\uparrow, \downarrow\}$. From the way we have defined the total order $\succ$, it follows that if $(k_1, \alpha_1) \succ (k_2, \alpha_2) \succ (k_3, \alpha_3)$, then $k_1 \geq k_3+1$. The property holds since $c_{(k_1, \alpha_1)} = (1+\epsilon)^{k_1}$ and $c_{(k_3, \alpha_3)} = (1+\epsilon)^{k_3}$.
\end{proof}

We are now ready to state the key observation which leads to the proof of~(\ref{eq:th:1}).

\begin{observation}
\label{ob:key} Consider an event where some nonzero units of potential are absorbed by some nodes, under scenario (b). Such an event occurs only if some node $x$ moves up or down one level from its current level $k$ (say). In the former case let $\gamma^* = (k, \uparrow)$, and in the latter case let $\gamma^* = (k, \downarrow)$. Let $\delta^* \geq 0$ denote the potential of type $\gamma^*$ {\em released} by $x$ during this event. For every type $\gamma \in \Gamma$, let $\delta_{\gamma} \geq 0$ denote the total potential of type $\gamma$ {\em absorbed} by all the nodes during this event. Then:
\begin{enumerate}
\item For every $\gamma \in \Gamma$, we have $\delta_{\gamma} > 0$ only if $\gamma^* \succ \gamma$.
\item $\sum_{\gamma \in \Gamma} \delta_{\gamma} \leq \delta^*$. 
\end{enumerate}
\end{observation}

\begin{proof}
We consider two  cases, depending on whether the node $x$ is moving up or moving down.

\medskip
\noindent {\em Case 1. The node $x$ is active and up-dirty, and it is moving up from level $k$ to level $k+1$.}

\smallskip
\noindent
Consider any node $v$ that absorbs nonzero units of potential at a level $t$ during the up-jump of $x$. Then part (1) of Lemma~\ref{lm:transfer:up:jump} implies that $v$ is down-dirty after the up-jump of $x$ and $t \leq k$. In other words, if the node $v$ absorbs nonzero units of potential of type $\gamma$, then  $\gamma$ is of the form $(t, \downarrow)$ for some $t \leq k$. Also, note that $\gamma^*$ is of the form $(k, \uparrow)$. From the way we have defined our total order $\succ$ on the set $\Gamma$, we infer that $\gamma^* \succ \gamma$ for any such $\gamma$. The concludes the proof of part (1) of the observation. Part (2) of the observation follows from part (2) of Lemma~\ref{lm:transfer:up:jump}.

\smallskip
\noindent {\em Case 2. The node $x$ is active and down-dirty, and it is moving down from level $k$ to level $k-1$.}

\smallskip
\noindent
Consider any node $v$ that absorbs nonzero units of potential at a level $t$ during the down-jump of $x$. Then part (1) of Lemma~\ref{lm:transfer:down:jump} implies that $t < k$. In other words, if the node $v$ absorbs nonzero units of potential of type $\gamma$, then  $\gamma$ is of the form $(t, \alpha)$ for some $\alpha \in \{\uparrow, \downarrow\}$ and some $t < k$. Also, note that $\gamma^*$ is of the form $(k, \downarrow)$. From the way we have defined our total order $\succ$ on the set $\Gamma$, we infer that $\gamma^* \succ \gamma$ for any such $\gamma$. The concludes the proof of part (1) of the observation.  Part (2) of the observation follows from part (2) of Lemma~\ref{lm:transfer:down:jump}.
\end{proof}

Property~\ref{prop} and Observation~\ref{ob:key} together give us a complete picture of the way the potential  stored by the nodes {\em flows within the system}. Specifically, there are two scenarios in which potential can be pumped into (i.e., absorbed by) the nodes. In scenario (a), potential gets pumped into the nodes {\em exogenously} by an adversary, due to the insertion or deletion of an edge in the graph. But, according to~(\ref{eq:th:0}), the total energy absorbed by the nodes under this scenario is already upper bounded by $O(\tau/\epsilon)$. On the other hand, in scenario (b), a node releases some $\delta^* \geq 0$ units of potential of type $\gamma^* \in \Gamma$ (say), and at the same time  some $0 \leq \delta \leq \delta^*$ units of potential get created. This newly created $\delta$ units of potential are then {\em split up} in some {\em chunks}, and  these chunks in turn get absorbed as potentials of (one or more) different types $\gamma$. We now note three key points about this process: (1) $\delta \leq \delta^*$. (2) If a chunk of this newly created $\delta$ units of potential gets absorbed as potential of type $\gamma$, then we must have $\gamma^* \succ \gamma$. (3) By Property~\ref{prop}, as we move down two hops in the total order $\succ$, the conversion rate between energy and potential {\em drops} at least by a factor of $(1+\epsilon)$. These three points together imply that the energy absorbed under scenario (b) is at most $\beta$ times the energy absorbed under scenario (a), where $\beta = 2 + 2 \cdot (1+\epsilon)^{-1} + 2 \cdot (1+\epsilon)^{-2} + \cdots = O(1/\epsilon)$. Hence, from~(\ref{eq:th:0}) we infer that the total energy absorbed by the nodes under scenario (b) is at most $O(1/\epsilon) \cdot O(\tau/\epsilon) = O(\tau/\epsilon^2)$. This concludes the proof of~(\ref{eq:th:1}).

\section{Bounding the amortized update time of our algorithm}
\label{sec:analyze:time}

We now describe the data structures that we use to implement our dynamic algorithm for maintaining a $(2+\epsilon)$-approximate minimum vertex cover. Each node $v \in V$ maintains its  level $\ell(v)$, its weight $W_v$, and the following doubly linked lists, which we refer to as the {\em neighborhood lists for $v$}.
\begin{itemize}
\item For every level $t > \ell(v)$, the set $E_v(t) = \{ (u, v) \in E : \ell(u) = t \}$ of its incident edges whose other endpoints are at level $t$.
\item The set $E_v^- = \{ (u, v ) \in E : \ell(u) \leq \ell(v) \}$ of its incident edges whose other endpoints are either at the same level as $v$ or at a level lower than $v$.
\end{itemize}
The set of edges incident on  a node  is partitioned into its neighborhood lists. Each node  also maintains the size of each of its neighborhood lists. Furthermore, each node $v \in V$ maintains the value of $\left(W_{v \rightarrow k} - W_{v \rightarrow k+1}\right) = \delta W_v(k)$ at every level $k \in \{ \ell(v), \ldots, L-1\}$. If a node $v$ is up-dirty, then using the values of $W_v$ and $\delta W_v(\ell(v))$, it can figure out in constant time whether or not it is active. On the other hand, if a node $v$ is down-dirty, then it can figure out if it is active or not only by looking up the value of $W_v$ (again in constant time). Finally, each edge $(u, v) \in E$ maintains two pointers -- one pointing to the position of  $(u, v)$ in the neighborhood lists of $v$ and the other pointing to the position of $(u,v)$ in the neighborhood lists of $u$. Using these pointers, we can insert or delete an edge from a neighborhood list in constant time.  

In Claim~\ref{cl:update:time:up} and Claim~\ref{cl:update:time:down}, we bound the time taken to update these data structures during a single iteration of the {\sc While} loop in Figure~\ref{fig:dirty}. From Corollary~\ref{cor:update:time:up} and Corollary~\ref{cor:update:time:down} we infer that the time spent on updating the relevant data structures is within a constant factor of the energy released by the nodes. This, along with Theorem~\ref{th:amortized:work}, implies the following result.

\begin{theorem}
\label{th:amortized:update:time}
There is a dynamic algorithm for maintaining a $(2+\epsilon)$-approximate minimum vertex cover. Starting from an empty graph $G = (V, E)$, the algorithm spends $O(\tau/\epsilon^2)$ total time to handle $\tau$ edge insertions/deletions. Hence, the amortized update time of the algorithm is $O(1/\epsilon^2)$.
\end{theorem}

\begin{claim}
\label{cl:update:time:up}
Consider a node $x$  moving up from level $k$ to  $k+1$ during an iteration of the {\sc While} loop in Figure~\ref{fig:dirty}. It takes $O(|N_x(0,k)|)$ time to update the relevant data structures during this event.
\end{claim}

\begin{proof}
To update the relevant data structures, we need to perform the following operations. 
\begin{enumerate}
\item {\em Inform} the other endpoint $v$ of every edge $(v, x) \in E_x^-$ that the node $x$ has moved up one level. The other endpoint $v$  deletes the edge $(x, v)$ from the list $E_v(k)$ and inserts it back into the list $E_v(k+1)$, and  then updates its weight $W_v$ and the value of $\delta W_v(k)$. This takes $O(|N_x(0,k)|)$ time.
\item In $O(1)$ time, update the weight $W_x$ by setting $W_x = W_x - |N_x(0,k)| \cdot \left( (1+\epsilon)^{-k} - (1+\epsilon)^{-k-1}\right)$. 
\item Add all the edges in  $E_x(k+1)$ to  $E_x^-$. This takes $O(1)$ time, since the neighborhood lists are implemented as doubly linked lists.
\item Update the level of $v$ by setting $\ell(v) = \ell(v) + 1$. This takes $O(1)$ time.
\end{enumerate}
Thus, it takes $O(|N_x(0,k)|)$ time to update all the relevant data structures during this event. 
\end{proof}

\begin{corollary}
\label{cor:update:time:up}
Consider an event where a node $x \in V$ is moving  up from level $k$ to  $(k+1)$ during an iteration of the {\sc While} loop in Figure~\ref{fig:dirty}. During this event, the time spent to update the relevant data structures is within a $O(1)$ factor of  the energy released by the node  $x$.
\end{corollary}

\begin{proof}
During this event, the only edges that change their levels are of the form $(x, v) \in E$ with $v \in N_x(0,k)$. Hence, the work done on the node $x$  is given by $|N_x(0,k)|$. By Lemma~\ref{lm:work:jump},  the node $x$ releases at least $(1+\epsilon)^{-1} \cdot |N_x(0,k)|$ units of energy during this event. The corollary now follows from Claim~\ref{cl:update:time:up}.
\end{proof}

\begin{claim}
\label{cl:update:time:down}
Consider a  node $x$  moving down from level $k$ to $k-1$ during an iteration of the {\sc While} loop in Figure~\ref{fig:dirty}. During this event, it takes $O(|N_x(0,k)|)$ time to update the  data structures.
\end{claim}

\begin{proof}
To update the data structures, we need to perform the following operations. 
\begin{enumerate}
\item {\em Inform} the other endpoint $v$ of every edge $(v, x) \in E_x^-$ that  $x$ has moved down one level. If $\ell(v) < k-1$, then the other endpoint $v$ deletes the edge $(x, v)$ from $E_v(k)$ and inserts it back into $E_v(k-1)$, and appropriately  updates its weight $W_v$ and the value of $\delta W_v(k-1)$. If $\ell(v) = k-1$, then the other endpoint $v$ deletes the edge $(x, v)$ from $E_v(k)$ and inserts it back into $E_v^-$, and appropriately  updates its weight $W_v$ and the value of $\delta W_v(k-1)$. This takes $O(|N_x(0,k)|)$ time.
\item For every edge  $(x, v) \in E_x^-$, if $\ell(v) = k$,  then remove the edge $(x, v)$ from $E_x^-$ and insert it back into $E_x(k)$. This takes $O(|N_x(0,k)|)$ time. At the end of these operations, the list $E_x^-$ consists only of the edges $(x, v) \in E$ with $u \in N_x(0,k-1)$.
\item Update the weight $W_x$ as follows: $W_x = W_x + |N_x(0,k-1)| \cdot \left( (1+\epsilon)^{-k+1} - (1+\epsilon)^{-k}\right)$, and also set $\delta W_x(k-1) = |N_x(0,k-1)| \cdot \left( (1+\epsilon)^{-k+1} - (1+\epsilon)^{-k} \right)$. This  can be done in $O(1)$ time since at this stage $|N_x(0,k-1)| = |E^-_x|$, and since  $x$ maintains the size of each neighborhood lists.
\end{enumerate}
Thus,  it takes $O(|N_x(0,k)|)$ time to update all the relevant data structures during this event. 
\end{proof}

\begin{corollary}
\label{cor:update:time:down}
Consider an event where a  node $x \in V$ is moving down from level $k$ to  $(k-1)$ during an iteration of the {\sc While} loop in Figure~\ref{fig:dirty}. During this event, the time spent to update the relevant data structures is within a $O(1)$ factor of  the energy released by  $x$ at level $k$.
\end{corollary}

\begin{proof}
Since the node $x$ is moving down from level $k$ to level $k-1$, we have $W_{v \rightarrow k} < 1-\epsilon$. Hence, from~(\ref{eq:potential:down:new}) we infer that the node $x$ releases at least $\epsilon$ units of potential at level $k$ during this event. Since the {\em conversion rate} between energy and potential is $(1+\epsilon)^k \cdot \epsilon^{-1}$ at level $k$, the node $x$ releases at least $\epsilon \cdot (1+\epsilon)^k \cdot \epsilon^{-1} = (1+\epsilon)^k$ units of energy during this event. 

Next,  note that each edge $(v, x) \in E$ with $v \in N_x(0,k)$ contributes $(1+\epsilon)^{-k}$ units of weight towards $W_{x \rightarrow k}$. Since $W_{x \rightarrow k} < 1-\epsilon \leq 1$, we have $|N_x(0,k)| \cdot (1+\epsilon)^{-k} \leq 1$, which implies that $|N_x(0,k)| \leq (1+\epsilon)^k$. Thus, from Claim~\ref{cl:update:time:down} we conclude that it takes $O((1+\epsilon)^k)$ time to update the relevant data structures during this event, and this is clearly within a $O(1)$ factor of the energy released by  $x$.
\end{proof}

\section{Proof of Lemma~\ref{lm:insertion:deletion}}
\label{sec:lm:insertion:deletion}
Without any loss of generality, suppose that $\ell(u) = i \leq \ell(v) = j$ just before the insertion/deletion of the edge $(u, v)$. We will prove the lemma by separately considering two possible cases. 
\subsection{Case 1: The edge $(u, v)$ is being deleted.}

\noindent While analyzing this case, we use the term {\em event} to refer to the deletion of the edge $(u,v)$. We emphasize that this event does not include the execution of the {\sc While} loop in Figure~\ref{fig:dirty}. In particular, this event does not change the level of any node. We consider four possible scenarios, depending on whether the endpoint $x \in \{u, v\}$ is up-dirty or down-dirty  {\em before} and/or {\em after} the event. We will show that under all these scenarios the total energy absorbed by any endpoint $x \in \{u, v\}$ due to this event is at most $3\epsilon^{-1}$. 

\medskip
\noindent  {\bf Scenario 1:} {\em The node $x \in \{u, v\}$ is  down-dirty both before and after the event.} In this scenario, from~(\ref{eq:potential:down:new}), we conclude that the node $x$ absorbs $w(u, v) = (1+\epsilon)^{-\max(\ell(u), \ell(v))} = (1+\epsilon)^{-j}$ units of potential at level $\ell(x)$, and absorbs/releases zero unit of potential at every other level. Hence, the total energy absorbed by the node $x \in \{u, v\}$ due to this event is  $= (1+\epsilon)^{-j} \cdot (1+\epsilon)^{\ell(x)}  \cdot \epsilon^{-1} \leq \epsilon^{-1}$.

\medskip
\noindent
{\bf Scenario 2:} {\em The node $x \in \{u, v\}$ is up-dirty both before and after the event.} For every level $0 \leq k \leq L$, we observe that the weight $W_{x \rightarrow k}$ decreases by $(1+\epsilon)^{-\max(k, j)}$ due to this event. In particular, for every level $0 \leq k < L$, the decrease in the value $W_{x \rightarrow k}$ is  at least the decrease in the value of $W_{x \rightarrow k+1}$. This observation, along with~(\ref{eq:potential:up:new}), implies that due to this event the value of $\Phi(x, k)$ can only decrease at every level $0 \leq k \leq L$. In particular, the node $x$ absorbs zero potential and zero energy due to this event.

\medskip
\noindent
{\bf Scenario 3:} {\em The node $x \in \{u, v\}$ is up-dirty before the event and down-dirty after the event.} From~(\ref{eq:potential:down:new}), we conclude that after the event we have $\Phi(x, k) = 0$ at every level $k \neq \ell(x)$ and $\Phi(x, \ell(x)) = 1 - W_{x \rightarrow \ell(x)} = 1 - W_x$. Accordingly, due to this event the node $x$ absorbs  zero potential and zero energy at every level $k \neq \ell(x)$. Furthermore, due to this event, the potential absorbed by the node $x$ at level $\ell(x)$  is upper bounded by the decrease in the weight $W_x$, which is turn is equal to the weight $w(u, v) = (1+\epsilon)^{-\max(\ell(u), \ell(v))} = (1+\epsilon)^{-j}$. Hence, the total energy absorbed by  $x$ due to this event is at most $(1+\epsilon)^{-j} \cdot (1+\epsilon)^{\ell(x)} \cdot \epsilon^{-1} \leq \epsilon^{-1}$. 
 
\medskip
\noindent
{\bf Scenario 4:} {\em The node $x \in \{u, v\}$ is down-dirty before the event and up-dirty after the event.} As the event  decreases the weight $W_x$, this scenario  never takes place.

\subsection{Case 2: The edge $(u, v)$ is being inserted.}

\noindent While analyzing this case, we use the term {\em event} to refer to the insertion of the edge $(u,v)$. We emphasize that this event does not include the execution of the {\sc While} loop in Figure~\ref{fig:dirty}. In particular, this event does not change the level of any node.  We consider four possible scenarios, depending on whether the endpoint $x \in \{u, v\}$ is up-dirty or down-dirty  {\em before} and/or {\em after} the event. We will show that under all these scenarios the total energy absorbed by any endpoint $x \in \{u, v\}$ due to this event is at most $3\epsilon^{-1}$. 

\medskip
\noindent {\bf Scenario 1:} {\em The node $x \in \{u, v\}$ is down-dirty both before and after the event.} From~(\ref{eq:potential:down:new}), we conclude that the node $x$ releases $w(u, v) = (1+\epsilon)^{-\max(\ell(u), \ell(v))} = (1+\epsilon)^{-j}$ units of potential at level $\ell(x)$, and absorbs/releases zero  potential at every other level. Thus, according to our convention,   the node $x$ absorbs zero  potential  at every level, which means that the node $x$ absorbs zero  energy overall due to this event.
 
\medskip
\noindent {\bf Scenario 2:} {\em The node $x \in \{u, v\}$ is up-dirty both before and after the event.} In this scenario, from Invariant~\ref{inv} we first recall that the node $x$ is passive before the event. Hence, by definition we have $W_{x \rightarrow k} \leq W_{x \rightarrow \ell(x)+1} < 1$ at every level $k > \ell(x)$ before the event.  From~(\ref{eq:potential:up:new}), we infer that:
\begin{eqnarray}
\label{eq:veryimp:1}
\text{Just before the event, we have } \Phi(x, k) = \begin{cases} 0 & \text{ at every level } k < \ell(x); \\ 
  W_x - 1 & \text{ at level } k = \ell(x); \\
  0 & \text{ at every level } k > \ell(x).
  \end{cases}
\end{eqnarray}
Just after the event, we can no longer claim that the node $x$ remains passive. Accordingly, let $k_x$ denote the largest level  $k \in \{\ell(x), \ldots, L\}$ where $W_{x \rightarrow k} \geq 1$ after the event. Note that $k_x$ is the value of $\ell(x, \uparrow)$ after the event. Hence, from~(\ref{eq:potential:up:new}) we infer that:
\begin{eqnarray}
\label{eq:veryimp:2}
\text{Just after the event, we have } \Phi(x, k) = \begin{cases} 0 & \text{ at every level } k < \ell(x); \\ 
  W_{k} - W_{k+1} & \text{ at level }  \ell(x) \leq k < k_x; \\
  W_{k_x} - 1 & \text{ at level } k = k_x; \\
  0 & \text{ at every level } k > k_x.
  \end{cases}
\end{eqnarray}
From~(\ref{eq:veryimp:1}) and~(\ref{eq:veryimp:2}), we infer that $\Phi(x, k) = 0$ at every level $k < \ell(x)$ and at every level $k > k_x$ both before and after the event. In other words, due to the event the node $x$ absorbs nonzero potential only at levels $\ell(x) \leq k \leq k_x$. From Corollary~\ref{cor:total:potential:up:new} we infer that the total potential absorbed by the node $x$ across all the levels is at most the increase in the value of $W_x$ due to this event, which in turn is equal to the weight $w(u, v) = (1+\epsilon)^{-\max(\ell(u), \ell(v))} = (1+\epsilon)^{-j}$. Since the node $x$ absorbs this potential  only at levels $\ell(x) \leq k \leq k_x$, the total energy absorbed by the node $x$ due to this event is upper bounded by $(1+\epsilon)^{-j} \cdot (1+\epsilon)^{k_x} \cdot \epsilon^{-1} \leq (1+\epsilon)^{-\ell(x)} \cdot (1+\epsilon)^{k_x} \cdot \epsilon^{-1}$. We  will show that:
\begin{equation}
\label{eq:veryimp:3}
(1+\epsilon)^{k_x - \ell(x)} \leq 3
\end{equation}
The above inequality implies that the total energy absorbed by the node $x$ due to this event is at most $(1+\epsilon)^{-\ell(x)} \cdot (1+\epsilon)^{k_x} \cdot \epsilon^{-1} \leq 3 \epsilon^{-1}$. Hence, it now remains to prove~(\ref{eq:veryimp:3}). Towards this end, we first observe that the proof trivially holds if $k_x = \ell(x)$. Thus, for the rest of the proof we assume that $k_x \geq \ell(x)+1$. Note that there must be at least one edge $(x, y) \in E$ with $\ell(y) \leq \ell(x)$ just before the event, for otherwise the value of $W_x$ will not change even if we increase the level of $x$ by one and hence the node $x$ will not be passive before the event (thereby contradicting Invariant~\ref{inv}). Now, consider any edge $(x, y) \in E$ with $\ell(y) \leq \ell(x)$. If we were to move the node $x$ up from level $\ell(x)+1$ to level $k_x$, then the value of $w(x, y)$ will decrease by $(1+\epsilon)^{-\ell(x)-1} - (1+\epsilon)^{-k_x}$. Further, no edge incident on $x$ would increase its weight if we were move the node $x$ up from one level to another. These two observations together imply that:
\begin{eqnarray}
\label{eq:veryimp:4}
W_{x \rightarrow k_x}   \leq W_{x \rightarrow \ell(x)+1}  - \left((1+\epsilon)^{-\ell(x)-1} - (1+\epsilon)^{-k_x}\right)    \text{ before the event.}
\end{eqnarray}
Since the node $x$ is up-dirty and passive just before the event (by  Invariant~\ref{inv}), we have $W_{x \rightarrow \ell(x)+1} < 1$ at that time.  Combining this observation with~(\ref{eq:veryimp:4}), we get:
\begin{eqnarray}
\label{eq:veryimp:5}
W_{x \rightarrow k_x} \leq 1 - \left((1+\epsilon)^{-\ell(x)-1} - (1+\epsilon)^{-k_x}\right) \text{ before the event.}
\end{eqnarray}
The event (insertion of an edge incident on $x$) increases the values of $W_{x \rightarrow k_x}$ by at most $(1+\epsilon)^{-k_x}$. Combining this observation with~(\ref{eq:veryimp:5}), we get:
\begin{eqnarray}
\label{eq:veryimp:6}
W_{x \rightarrow k_x} \leq 1 - \left((1+\epsilon)^{-\ell(x)-1} - (1+\epsilon)^{-k_x}\right) + (1+\epsilon)^{-k_x} \text{ after the event.}
\end{eqnarray}
By definition of the level $k_x$, we have $1 \leq W_{x \rightarrow k_x}$ after the event. Hence, from~(\ref{eq:veryimp:6}) we infer that:
$$1 \leq 1 + 2 (1+\epsilon)^{-k_x} - (1+\epsilon)^{-\ell(x)-1}.$$ 
Rearranging the terms in the above inequality, we get $(1+\epsilon)^{-\ell(x)-1} \leq  2 (1+\epsilon)^{-k_x}$, which implies that $(1+\epsilon)^{k_x-\ell(x)} \leq  2(1+\epsilon) \leq 3$. This concludes the proof of~(\ref{eq:veryimp:3}).

\medskip
\noindent {\bf Scenario 3:} {\em The node $x \in \{u, v\}$ is down-dirty before the event and up-dirty after the event.}  The argument here is very similar to the argument under Scenario 2 and is therefore omitted.

\medskip
\noindent
{\bf Scenario 4:} {\em The node $x \in \{u, v\}$ is up-dirty before the event and down-dirty after the event.} As the event (the insertion of an edge incident on $x$) increases  $W_x$, this scenario never takes place.

\section{Proof of Lemma~\ref{lm:transfer:up:jump}}
\label{sec:lm:transfer:up:jump}

We prove part (1) and part (2) of the lemma respectively in Section~\ref{sec:lm:transfer:up:jump:part1} and in Section~\ref{sec:lm:transfer:up:jump:part2}.

\subsection{Proof of part (1) of Lemma~\ref{lm:transfer:up:jump}}
\label{sec:lm:transfer:up:jump:part1}

The node $x$ is active and up-dirty just before it moves up from level $k$ to level $k+1$. By Corollary~\ref{cor:active}, we have $W_{x \rightarrow k} \geq W_{x \rightarrow k+1} \geq 1$. From~(\ref{eq:potential:up:new}), we infer that during its up-jump from level $k$ to level $k+1$, the node $x$ releases $\left(W_{x \rightarrow k} - W_{x \rightarrow k+1}\right)$ units of potential  at level $k$ and absorbs/releases zero unit of potential at every other level $t \neq k$. Thus, we have $\delta(x, t, \leftarrow) = 0$ at every level $0 \leq t \leq L$. In particular, the node $x$ satisfies the condition stated in part (1) of Lemma~\ref{lm:transfer:up:jump}.

Next, consider any node $v \notin  \{x \} \cup N_x$. For such a node $v$, note that the value of $W_{v \rightarrow t}$ remains unchanged for all $0 \leq t \leq L$ as the node $x$ changes its level. Hence, from~(\ref{eq:potential:up:new}) and~(\ref{eq:potential:down:new}) we infer that $\delta(v, t, \leftarrow) = 0$ for all $0 \leq t \leq L$. In particular, such a node $v$ satisfies part (1) of Lemma~\ref{lm:transfer:up:jump}.

For the rest of the proof, we fix a node $v \in N_x$ and consider two possible cases depending on the state of $v$ after the up-jump of $x$. We show that in each of these cases the node $v$  satisfies part (1) of Lemma~\ref{lm:transfer:up:jump}.

\medskip
\noindent {\bf Case 1.} {\em The node $v$ is up-dirty after the up-jump of $x$.}

\medskip
\noindent In this case,  as the node $x$ moves up from level $k$ to level $k+1$, it can only decrease the weights $w(x,v)$ and $W_v$. Hence, the node $v$ must also have been up-dirty just before the up-jump of $x$. Next, we make the following observations. Consider the up-jump of $x$ from level $k$ to level $k+1$. Due to this event:
\begin{itemize}
\item  The value of $W_{v \rightarrow t}$ remains unchanged for every $t \in \{k+1, \ldots, L\}$.
\item  The value of $W_{v \rightarrow t}$ decreases by exactly the same amount for every $t \in \{0, \ldots, k\}$.
\end{itemize}
We now fork into two possible sub-cases depending on the level of $v$. 

\medskip
\noindent {\em Case 1(a). We have $\ell(v) \geq k+1$.} Here, the above two observations, along with~(\ref{eq:potential:up:new}), imply that during the up-jump of $x$ the node $v$ absorbs/releases zero unit of potential at every level $0 \leq t \leq L$.

\medskip
\noindent {\em Case 1(b). We have $\ell(v) \leq k$.} Here, the above two observations, along with~(\ref{eq:potential:up:new}), imply that during the up-jump of $x$ the node $v$ absorbs zero unit of potential at every level $0 \leq t \leq L$. But unlike in Case 1(a), here it is possible that the node $v$ might {\em release} nonzero units of potentials at level $t = k$.

\medskip
\noindent
In both Case 1(a) and Case 1(b), we derived  that $\delta(v, t, \leftarrow) = 0$ at every level $t \in \{0, \ldots, L\}$. In particular, this means that  node $v$ satisfies part (1) of Lemma~\ref{lm:transfer:up:jump}.

\medskip
\noindent {\bf Case 2.} {\em The node $v$ is down-dirty after the up-jump of $x$.}

\medskip
\noindent
Consider any level $0 \leq t \leq L$ such that $\delta(v, t, \leftarrow) > 0$.  Since the node $v$ does not change its own level during the up-jump of $x$ and since it is down-dirty after the up-jump of $x$, from~(\ref{eq:potential:down:new}) we conclude that $t = \ell(v)$. Next,   if it were the case that  $t = \ell(v) \geq k+1$, then the value of $W_{v \rightarrow t}$  would remain unchanged as the node $x$ moves up from level $k$ to level $k+1$, and since the node $v$ is down-dirty after the up-jump of $x$, from~(\ref{eq:potential:down:new}) we would be able to derive that $\delta(v, t, \leftarrow) = 0$.  Accordingly, we get $k \geq t = \ell(v)$, and derive that the node $v$ satisfies part (1) of Lemma~\ref{lm:transfer:up:jump}.

\subsection{Proof of part (2) of Lemma~\ref{lm:transfer:up:jump}}
\label{sec:lm:transfer:up:jump:part2}

Let $A_x \subseteq N_x(0,k)$ denote the subset of neighbors $v$ of $x$ such that  $\ell(v) \leq k$ and  $v$ is down-dirty immediately after the up-jump of $x$ under consideration. Part (1) of  Lemma~\ref{lm:transfer:up:jump} implies that $A_x$ is precisely the set of nodes $v$ with $\delta(v, \leftarrow) = \sum_{t=0}^L \delta(v,t, \leftarrow) > 0$. Note that  for all nodes $v \in A_x$, each of the weights $W_v$ and $w(x, v)$ decreases  by  $(1+\epsilon)^{-k} - (1+\epsilon)^{-(k+1)}$ as  $x$ moves up from level $k$ to level $k+1$. Since the node $v \in A_x$ is down-dirty after the up-jump of $x$, from~(\ref{eq:potential:down:new}) we infer that $\delta(v, \leftarrow) = \delta(v, \ell(v), \leftarrow) \leq (1+\epsilon)^{-k} - (1+\epsilon)^{-(k+1)}$. Using this observation, we now get:
\begin{eqnarray}
\nonumber 
\sum_{v, t} \delta(v, t, \leftarrow) & = & \sum_{v \in A_x} \delta(v, \leftarrow)  \leq  |A_x| \cdot  \left( (1+\epsilon)^{-k} - (1+\epsilon)^{-(k+1)}\right) \nonumber  \\
& \leq & |N_x(0,k)| \cdot \left( (1+\epsilon)^{-k} - (1+\epsilon)^{-(k+1)}\right)  =  W_{x \rightarrow k} - W_{x \rightarrow k+1} \label{eq:imp:1}
\end{eqnarray}
In the above derivation, the last step holds for the following reason. As the node $x$ moves up from level $k$ to level $k+1$, the only edges incident to $x$ that change their weights are the ones whose other endpoints belong to $N_x(0, k)$. Furthermore, each such edge $(x, v) \in E$ with $v \in N_x(0,k)$ changes its weight from $(1+\epsilon)^{-k}$ to $(1+\epsilon)^{-(k+1)}$ as the node $x$ moves up from level $k$ to level $k+1$. 

Since the node $x$ was an active and up-dirty node at level $k$ immediately before moving up to level $k+1$, by Corollary~\ref{cor:active} we have $W_{x \rightarrow k} \geq W_{x \rightarrow k+1} \geq 1$. Thus, from~(\ref{eq:potential:up:new}) we infer that the  potential released by $x$ at level $k$ during its up-jump is equal to $\delta(x, k, \rightarrow) = W_{v \rightarrow k} -  W_{v \rightarrow k+1}$. Hence, using~(\ref{eq:imp:1}) we now get   $\sum_{v, t} \delta(v, t, \leftarrow) \leq \delta(x, k, \rightarrow)$.  This concludes the proof.

\section{Proof of Lemma~\ref{lm:transfer:down:jump}}
\label{sec:lm:transfer:down:jump}

We prove part (1) and part (2) of the lemma respectively in Section~\ref{sec:lm:transfer:down:jump:part1} and in Section~\ref{sec:lm:transfer:down:jump:part2}.

\subsection{Proof of part (1) of Lemma~\ref{lm:transfer:down:jump}}
\label{sec:lm:transfer:down:jump:part1}

Since the node $x$ is {\em active} and down-dirty just before it moves down from level $k$ to level $k-1$,  Corollary~\ref{cor:active} implies that $W_{x \rightarrow k} \leq W_{x \rightarrow k-1} < 1$. From~(\ref{eq:potential:down:new}), we infer that while moving from level $k$ to level $k-1$, the node $x$ {\em releases} $\left(1 - W_{x \rightarrow k}\right)$ units of potential at level $k$ and absorbs $(1-W_{x \rightarrow k-1})$ units of potential at level $k-1$. Furthermore, the node $x$ releases/absorbs zero potential at every other level during this event. We conclude that $\delta(x, t, \leftarrow) > 0$ only if $t = k-1$. In particular, the node $x$ satisfies part (1) of Lemma~\ref{lm:transfer:down:jump}. 
  
 Next, consider any node $v \notin  \{x \} \cup N_x$. For such a node $v$, note that the value of $W_{v \rightarrow t}$ remains unchanged for all $0 \leq t \leq L$ as the node $x$ changes its level. Hence, from~(\ref{eq:potential:up:new}) and~(\ref{eq:potential:down:new}) we infer that $\delta(v, t, \leftarrow) = 0$ for all $0 \leq t \leq L$. In particular, such a node $v$ satisfies part (1) of Lemma~\ref{lm:transfer:down:jump}.
  
For the rest of the proof, we fix a node $v \in N_x$. We now consider four possible cases depending on the state and the level of $v$ just before the down-jump of $x$, and we show that in each of these cases the node $v$ satisfies part (1) of Lemma~\ref{lm:transfer:down:jump}.

\medskip
\noindent {\em Case 1. The node $v$ is up-dirty and $\ell(v) \geq k$ just before the down-jump of $x$.} 

\medskip
\noindent
In this case, for every level $t \geq k$ the value of $W_{v \rightarrow t}$ remains unchanged during the down-jump of $x$. Hence, from~(\ref{eq:potential:up:new}) we infer that $\delta(v,t, \leftarrow) = 0$ at every level $0 \leq t \leq L$. In particular, such a node $v$ satisfies part (1) of Lemma~\ref{lm:transfer:down:jump}.

\medskip
\noindent {\em Case 2. The node $v$ is up-dirty and $\ell(v) < k$ just before the down-jump of $x$.} 

\medskip
\noindent 
In this case, we first note that the node $v$ remains up-dirty just after the down-jump of $x$, since the weights  $w(x, v)$ and $W_v$ can only increase as the node $x$ moves down one level. We also note that the value of $W_{v \rightarrow t}$ remains unchanged at every level $t \geq k > \ell(v)$ during the down-jump of $x$. Hence, from~(\ref{eq:potential:up:new})  we infer that $\delta(v, t, \leftarrow) = 0$ at every level $t \geq k$. In particular, such a node $v$ satisfies part (1) of Lemma~\ref{lm:transfer:down:jump}.

\medskip
\noindent {\em Case 3. The node $v$ is down-dirty and $\ell(v) \geq k$ just before the down-jump of $x$.} 

\medskip
\noindent In this case, the value of $W_v$ remains unchanged as the node $x$ moves down from level $k$ to level $k-1$. Hence, from~(\ref{eq:potential:down:new}) we infer that $\delta(v, t, \leftarrow) = 0$ at every level $0 \leq t \leq L$. In particular, such a node $v$ satisfies part (1) of Lemma~\ref{lm:transfer:down:jump}.

\medskip
\noindent {\em Case 4. The node $v$ is down-dirty and $\ell(v) < k$ just before the down-jump of $x$.} 

\medskip
\noindent We consider two possible sub-cases, depending on the state of $v$ after the down-jump of $x$.

\medskip
\noindent {\em Case 4(a). The node $v$ is down-dirty before and up-dirty after the down-jump of $x$, and $\ell(v) < k$.}

\medskip
\noindent In this sub-case, we note that the value of $W_{v \rightarrow t}$ remains unchanged at every level $t \geq k$ as the node $x$ moves down from level $k$ to level $k-1$. Also, note that just before the down-jump of $x$, at every level $t \geq k > \ell(v)$ we have $W_{v \rightarrow t} \leq W_{v \rightarrow \ell(v)} < 1$. These two observations together imply that at every level $t \geq k > \ell(v)$, we have $W_{v \rightarrow t} < 1$ both before and after the down-jump of $x$. Since the node $v$ is down-dirty before and up-dirty after the down-jump of $x$, from~(\ref{eq:potential:up:new}) we infer that $\delta(v, t, \leftarrow) = 0$ at every level $t \geq k$. Thus, such a node $v$ satisfies part (1) of Lemma~\ref{lm:transfer:down:jump}.

\medskip
\noindent {\em Case 4(b). The node $v$ is down-dirty both before and after the down-jump of $x$, and $\ell(v) < k$.}

\medskip
\noindent In this sub-case, as the node $x$ moves down from level $k$ to level $k-1$, the weights $w(u,v)$ and $W_v = W_{v \rightarrow \ell(v)}$ increase. However, since the node $v$ remains down-dirty even after the down-jump of $x$, we conclude that the weight $W_v$ does not exceed $1$ during this event. From~(\ref{eq:potential:down:new}), we now infer that the node $v$ absorbs nonzero units of potential at level $\ell(v) < k$ and absorbs/releases zero unit of potential at every other level. Specifically, we have $\delta(v, t, \rightarrow) > 0$ at level $t = \ell(v) < k$ and $\delta(v, t, \rightarrow) = 0$ at every other level $t \neq \ell(v)$.  Such a node $v$ clearly satisfies part (1) of Lemma~\ref{lm:transfer:down:jump}.

\subsection{Proof of part (2) of Lemma~\ref{lm:transfer:down:jump}}
\label{sec:lm:transfer:down:jump:part2}

For every node $v \in V$, define $\delta(v, \leftarrow) = \sum_{t=0}^L \delta(v, t, \leftarrow)$ to be the total potential absorbed by $v$  during the down-jump of $x$. We first show that every node $v \in V \setminus \{x\}$ with $\delta(v, \leftarrow) > 0$ must be at a level $\ell(v) < k$.

\begin{claim}
\label{cl:lm:transfer:down:jump:1}
For every node $v \in V \setminus \{x\}$, we have $\delta(v, \leftarrow) > 0$ only if $v \in N_x(0, k-1)$.
\end{claim}

\begin{proof}
Consider any node $v \in V \setminus \{x\}$ with $\delta(v, \leftarrow) > 0$. Since $\delta(v, \leftarrow) = \sum_{t=0}^L \delta(v, t, \leftarrow)$, there must be a level $0 \leq t' \leq L$ such that $\delta(v, t', \leftarrow) > 0$. Part (1) of Lemma~\ref{lm:transfer:down:jump} implies that $k > t'$. To prove the claim, all we need to show is that $t' \geq \ell(v)$. To see why this inequality holds, for the sake of contradiction suppose that $\ell(v) > t'$. Then from~(\ref{eq:potential:up:new}) and~(\ref{eq:potential:down:new}) we get $\Phi(v, t') = 0$ both before and after the down-jump of $x$. This implies that $\delta(v, t',\leftarrow) = 0$, thereby leading to a  contradiction.
\end{proof}

For notational convenience,  for the rest of this section we define  $\lambda_{(x, v)} = (1+\epsilon)^{-(k-1)} - (1+\epsilon)^{-k}$ for every node $v \in N_x(0, k-1)$. Now, consider any such node $v \in N_x(0, k-1)$ with $\delta(v, \leftarrow) > 0$.  As the node $x$ moves down from level $k$ to level $k-1$, the weight  $w(x, v)$ increases by  $\lambda_{(x,v)}$.   Since the node $x$ was active and down-dirty just before making the down-jump from level $k$,  by Corollary~\ref{cor:active} we have $W_{x \rightarrow k} \leq W_{x \rightarrow k-1} < 1$.  Hence, from~(\ref{eq:potential:down:new}) we infer that the  potential released by $x$ at level $k$ during its down-jump is equal to  $\delta(x, k, \rightarrow) = 1 - W_{x \rightarrow k}$. During the same down-jump, the  potential absorbed by   the node  $x$ at level $k-1$ is equal to $\delta(x, k-1, \leftarrow) = 1 - W_{x \rightarrow k-1}$. Thus, we get: 
\begin{equation}
\label{eq:imp:11}
\delta(x, k, \rightarrow) - \delta(x, k-1, \leftarrow) = W_{x \rightarrow k-1} - W_{x \rightarrow k} = \sum_{v \in N_x(0, k-1)} \lambda_{(x, v)}
\end{equation} 
In the above derivation, the last equality holds since  only the edges $(x, v) \in E$ with $v \in N_x(0,k-1)$ change their weights as $x$ moves down from level $k$ to level $k-1$, and since $W_x = \sum_{v \in N_x} w(x, v)$. 

Since the node $x$ is active and down-dirty just before its down-jump from level $k$ to  $k-1$, by Corollary~\ref{cor:active} we have $W_{x \rightarrow k} \leq W_{x \rightarrow k-1} < 1$.  Thus, from~(\ref{eq:potential:down:new}) we infer  that while moving down from level $k$ to level $k-1$  the node $x$ absorbs nonzero units of potential only at level $k-1$. Specifically, we have $\delta(x, \leftarrow) = \sum_{t} \delta(x, t, \leftarrow) = \delta(x, k-1, \leftarrow)$. Hence, from~(\ref{eq:imp:11}) we get:
\begin{equation}
\label{eq:imp:2}
\delta(x, k, \rightarrow) - \delta(x, \leftarrow)  = \sum_{v \in N_x(0, k-1)} \lambda_{(x, v)}
\end{equation}
We will prove the following inequality.
\begin{equation}
\label{eq:cl:burn:down:1}
\lambda_{(x,v)} \geq \delta(v, \leftarrow) \text{ for every node } v \in N_x(0,k-1).
\end{equation}
Note that if we sum~(\ref{eq:cl:burn:down:1}) over all the nodes in $N_x(0,k-1)$, then~(\ref{eq:imp:2}) and Claim~\ref{cl:lm:transfer:down:jump:1} would give us:
$$\delta(x, k, \rightarrow) - \delta(x,  \leftarrow) = \sum_{v \in N_x(0,k-1)} \lambda_{(x,v)}  \geq \sum_{v \in N_x(0,k-1)} \delta(v, \leftarrow) = \sum_{v \in N_x: \delta(v, \leftarrow) > 0}  \delta(v, \leftarrow) = \sum_{v \in N_x} \delta(v, \leftarrow)$$
Rearranging the terms in the above inequality, we would get: 
$$\delta(x, k, \rightarrow) \leq \delta(x, \leftarrow) + \sum_{v \in N_x} \delta(v, \leftarrow) = \sum_{v \in N_x \cup \{x\}} \delta(v, \leftarrow) =\sum_{v \in N_x \cup \{x\}} \sum_{t=0}^L \delta(v, t, \leftarrow)  = \sum_{v, t} \delta(v, t, \leftarrow).$$ 
In the above derivation, the last step follows from part (1) of Lemma~\ref{lm:transfer:down:jump}. Note that this derivation leads us to  the proof of part (2) of Lemma~\ref{lm:transfer:down:jump}. Thus, it remains to show  why inequality~(\ref{eq:cl:burn:down:1}) holds, and accordingly, we focus on justifying inequality~(\ref{eq:cl:burn:down:1}) for the rest  this section. Towards this end, we fix any node $v \in N_x(0,k-1)$, and consider three possible cases. We show that the node $v$ satisfies inequality~(\ref{eq:cl:burn:down:1}) in each of these cases.

\subsubsection{Case 1:  $v$ is down-dirty  after the down-jump of $x$.} In this case, we infer that the node $v$ was down-dirty immediately before the down-jump of $x$ as well, for as the node $x$ moves down from level $k$ to level $k-1$, it can only increase the weight $W_v$.  Since the down-jump of $x$ can only increase the weight $W_v$, and since $v$ remains down-dirty both before and after the down-jump of $x$, from~(\ref{eq:potential:down:new}) we infer that $\delta(v, \leftarrow) = 0 \leq \lambda_{(x,v)}$.

\subsubsection{Case 2:  $v$ is up-dirty  before and up-dirty after the down-jump of $x$.} For notational convenience, we use the symbols $\tau_0$ and $\tau_1$ to respectively denote the time-instants immediately before and immediately after the down-jump of $x$. So the node $v$ is up-dirty both at time $\tau_0$ and at time $\tau_1$. Recall that $\ell(v) < k$. We  bound the sum $\delta(v, \leftarrow) = \sum_{t = 0}^{L} \delta(v, t, \leftarrow)$ by considering four possible ranges of values for $t$.
\begin{itemize}
\item (I) $0 \leq t < \ell(v)$. Here, from~(\ref{eq:potential:up:new}) we infer that $\Phi(v, t) = 0$ both at time $\tau_0$ and at time $\tau_1$. Hence, it follows that $\delta(v, t, \leftarrow) = 0$.
\item (II) $\ell(v) \leq t < k-1$. Here, we observe that the values of $W_{v \rightarrow t}$ and $W_{v \rightarrow t+1}$ {\em increase} by exactly the same amount as the node $x$ moves down from level $k$ to level $k-1$. Thus, from~(\ref{eq:potential:up:new}) we derive that $\delta(v, t, \leftarrow) = 0$. 
\item (III) $t=k-1$. Here, the value of $W_{v \rightarrow t}$ increases {\em exactly} by $\lambda_{(x, v)}$ as the node $x$ moves down from level $k$ to level $k-1$, but the value of $W_{v \rightarrow t+1}$ remains unchanged.  Thus, from~(\ref{eq:potential:up:new}) we derive that $\delta(v, t, \leftarrow) \leq \lambda_{(x,v)}$.
\item (IV) $k \leq t < L$. Here,  the values of $W_{v \rightarrow t}$ and $W_{v \rightarrow t+1}$ do not change as the node $x$ moves down from level $k$ to level $k-1$.  Thus, from~(\ref{eq:potential:up:new}) we derive that  $\delta(v, t, \leftarrow) = 0$.
\end{itemize}
To summarize, the above four observations imply that $\delta(v, \leftarrow) = \sum_{t=0}^L \delta(v, t, \leftarrow) \leq \lambda_{(x, v)}$.

\subsubsection{Case 3:   $v$ is down-dirty   before and up-dirty  after the down-jump of $x$.}

As in Case 2, for notational convenience we use the symbols $\tau_0$ and $\tau_1$ to respectively denote the time-instants immediately before and immediately after the down-jump of $x$. So the node $v$ is down-dirty  at time $\tau_0$ and up-dirty at time $\tau_1$. 

We use the superscripts $(0)$ and $(1)$  to respectively denote the values of any quantity at times $\tau_0$ and $\tau_1$. Consider the down-jump of $x$ from level $k$ to level $k-1$. Since the node $v$ is down-dirty before and up-dirty after this event,  it absorbs $\Phi^{(1)}(v, t)$ units of potential at every level $t$ due to this event. Hence, we get: 
\begin{equation}
\label{eq:cl:burn:down:7}
\delta(v, \leftarrow) = \sum_{t = 0}^{L} \delta(v, t, \leftarrow) = \sum_{t=0}^L \Phi^{(1)}(v, t) = \Phi^{(1)}(v) = W^{(1)}_{v} - 1 < W^{(1)}(v) - W^{(0)}(v).
\end{equation}
In the above derivation, the third step  follows from Corollary~\ref{cor:total:potential:up:new} and the fact that the node $v$ is up-dirty at time $\tau_1$. The last step holds  since the node $v$ is down-dirty at time $\tau_0$. In words, the value of $\delta(v, \leftarrow)$ is at most the increase in the weight  $W_v$ due to the down-jump of $x$. Note that the latter quantity is  equal to the increase in the weight $w(x, v)$ due to the down-jump of $x$, which in turn  is equal to $\lambda_{(x, v)}$ by definition. Thus, we conclude that $\delta(v, \leftarrow) \leq \lambda_{(x, v)}$.

\end{document}